\documentclass[11pt,reqno]{amsart}
\usepackage[top=2.54cm, bottom=2.54cm, left=2.8cm, right=2.8cm]{
geometry}
\usepackage{dsfont, amssymb,amsmath,amscd,latexsym, amsthm, amsxtra,amsfonts, extarrows}
\usepackage[all]{xy}
\usepackage[active]{srcltx}
\usepackage[round]{natbib}
\usepackage{bbm}
\usepackage{enumerate}
\usepackage{mathrsfs}
\usepackage{graphicx}
\usepackage{comment}
\usepackage{mathtools}

\usepackage{tikz}
\usetikzlibrary{calc,arrows}
\usepackage{verbatim}
\usepackage{graphicx}
\usepackage{subfigure}
\usepackage{units}
\usepackage{float}

\newtheorem{theorem}{Theorem}[section]

\theoremstyle{definition}
\newtheorem{definition}[theorem]{Definition}

\renewcommand{\theequation}{\arabic{section}.\arabic{equation}}

\theoremstyle{definition}

\theoremstyle{definition}
\newtheorem{remark}{Remark}
\theoremstyle{definition}

\newcommand{\rd}{\mathrm{d}}

\renewcommand{\epsilon}{\varepsilon}

\usepackage[pdfstartview=FitH, bookmarksnumbered=true,bookmarksopen=true, colorlinks=true, pdfborder=001, citecolor=blue, linkcolor=blue,urlcolor=blue]{hyperref}
\usepackage{graphics}
\graphicspath{{figures/}}
\usepackage{xpatch}
\usepackage{xcolor}

\makeatletter
\ExplSyntaxOn
\cs_new:Npn \bibColoredItems #1#2
{
	\clist_map_inline:nn {#2} { \cs_new:cpn {bib@colored@##1} {#1} }
}
\ExplSyntaxOff

\newcommand\bib@setcolor[1]{%
	\ifcsname bib@colored@#1\endcsname
	\expandafter\color\expandafter{\csname bib@colored@#1\endcsname}
	\else
	\normalcolor
	\fi
}

\xpatchcmd\@bibitem
{\item}
{\bib@setcolor{#1}\item}
{}{\fail}

\xpatchcmd\@lbibitem
{\item}
{\bib@setcolor{#2}\item}
{}{\fail}
\makeatother

\begin{document}
	\makeatletter
	\def\@setauthors{%
		\begingroup
		\def\thanks{\protect\thanks@warning}%
		\trivlist \centering\footnotesize \@topsep30\p@\relax
		\advance\@topsep by -\baselineskip
		\item\relax
		\author@andify\authors
		\def\\{\protect\linebreak}%
		{\authors}%
		\ifx\@empty\contribs \else ,\penalty-3 \space \@setcontribs
		\@closetoccontribs \fi
		\endtrivlist
		\endgroup} \makeatother
	\baselineskip 18pt
	\title[{{\tiny $N$-player and mean field games}}]
	{{\tiny
			$N$-player and mean field games among fund managers considering excess logarithmic returns}} \vskip 10pt\noindent
	\author[{\tiny  Guohui Guan, Jiaqi Hu, Zongxia Liang}]
	{\tiny {\tiny  Guohui Guan$^{a,b,*}$, Jiaqi Hu$^{c,\dag}$, Zongxia Liang$^{c,\ddag}$}
		\vskip 10pt\noindent
		{\tiny ${}^a$Center for Applied Statistics, Renmin University of China, Beijing 100872, China
			\vskip 10pt\noindent\tiny
			${}^b$School of Statistics, Renmin University of China, Beijing 100872, China
			\vskip 10pt\noindent
			${}^c$Department of Mathematical Sciences, Tsinghua
			University, Beijing 100084, China
		}\noindent
		\footnote{{$^*$ {\bf e-mail}: guangh@ruc.edu.cn}}\noindent
		\footnote{{$^\dag $ Corresponding author, \  {\bf e-mail}: hujq20@mails.tsinghua.edu.cn }}  \noindent
		\footnote{{$\ddag $ {\bf e-mail}: liangzongxia@tsinghua.edu.cn}}}
	\numberwithin{equation}{section}
	\maketitle
	\noindent
	\begin{abstract}
		This paper studies the competition among multiple fund managers with relative performance over the excess logarithmic return. Fund managers compete with each other and have expected utility or mean-variance criteria for excess logarithmic return. 
  Each fund manager possesses a unique risky asset, and all fund managers can also invest in a public risk-free asset and a public risk asset. We construct both an $n$-player game and a mean field game (MFG) to address the competition problem under these two criteria. We explicitly define and rigorously solve the equilibrium and mean field equilibrium (MFE) for each criteria. In the four models, the excess logarithmic return as the evaluation criterion of the fund leads to the { allocation fractions} being constant. The introduction of the public risky asset yields different outcomes, with competition primarily affecting the investment in public assets, particularly evident in the MFG.  We demonstrate that the MFE of the MFG represents the limit of the $n$-player game's equilibrium as the competitive scale $n$ approaches infinity. Finally, the sensitivity analyses of the equilibrium are given.
	\end{abstract}
	\vskip10pt
	\setcounter{equation}{0}

\section{{{\bf Introduction}}}

Portfolio selection has a long history in finance, leading to two significant research directions based on different criteria. One direction is derived from the mean-variance criterion, introduced by \cite{markowits1952portfolio} in a single-period setting and extended to a continuous-time setting by \cite{zhou2000continuous}. The other direction is based on expected utility maximization, as developed by \cite{Merton1969} and \cite{merton1975optimum}. These two different criteria have been widely applied in financial mathematics, see \cite{goldfarb2003robust}, \cite{liu2007portfolio}, \cite{yu2011investor}, \cite{chiu2016supply}, \cite{georgantas2024robust}, etc.

Most of the literature on portfolio selection focuses on absolute capital value, although some portfolio optimization studies consider return. In fact, it is more common in the broader financial literature to emphasize return rather than absolute capital value. \cite{markowits1952portfolio} conducts a mean-variance analysis of returns. However, in the one-period model, returns and capital values are equivalent, so subsequent research on the mean-variance criterion has focused on maximizing absolute capital values at the terminal time. \cite{kelly1956new} considers the maximization of the expected growth rate (logarithmic return) and introduced the concept of the growth optimal portfolio (GOP). There is also some follow-up literature, such as \cite{aurell2000general} and \cite{thorp2008kelly}. Many studies in finance that focus on returns use logarithmic returns, as does \cite{kelly1956new}, due to their additivity and the fact that they tend to follow a normal distribution.

In practice, institutions with long-term investment goals, particularly asset management funds, do not focus solely on the value of their portfolio capital. Instead, they prioritize portfolio returns and often assess performance using metrics like excess returns over stock indices or interest rates. For example, many of China's quant funds consider excess returns of the CSI 300, CSI 500, or CSI 1000, and China's industry-specific funds often compare their absolute return curve with the average return curve of the same industry. \cite{fama1993common} emphasize the importance of excess returns, providing an income-based perspective to understand portfolio performance.  \cite{malkiel1995returns} suggests that revenue-based assessments are more effective in the long run than capital growth alone. Using returns as a performance evaluation metric offers several practical advantages over capital value, as it can mitigate the impact of initial funding and investment duration, facilitating comparisons between different portfolio types. Investors' subscription and redemption behavior during the life of a fund can affect its size but not its return rate, allowing returns to eliminate this impact on fund evaluation. Moreover, funds are predominantly ranked based on their returns over a fixed period, leading to an influx of investments into funds that outperform benchmark indexes or rank at the top. For fund managers, achieving higher returns not only garners recognition and prestige but also attracts potential investors, contributing to fund expansion and increased management fees. In the fiercely competitive fund management landscape, evaluating portfolio performance based on returns rather than capital value appears more appropriate. Much of the literature on evaluating fund performance by  returns, such as \cite{carhart1997persistence}, \cite{cremers2009active} and \cite{berk2015measuring}, etc.

Given the importance of returns in the financial industry, recent literature on portfolio optimization has increasingly considered logarithmic returns. \cite{dai2021dynamic} proposes a dynamic portfolio choice model with the mean-variance criterion for logarithmic returns, whose portfolio policies conform with conventional investment wisdom. \cite{peng2023relative} considers the cumulative prospect theory and extends the classical growth optimal problem to the behavioral framework. Our research focuses on utilizing logarithmic returns to analyze the competitive dynamics among fund managers. In a competitive environment with multiple fund managers, excess returns are more significant than simple returns. Fund managers strive to outperform their peers or a benchmark that represents peer performance, aiming for a larger share of excess returns while being highly averse to underperforming the benchmark due to concerns about fund survival. Therefore, it is advisable to initially consider logarithmic returns and then assess the excess logarithmic return. \cite{peng2023relative} employs excess logarithmic return, where the benchmark can be chosen as the return of a stock index or structured product, the deposit or loan interest rate, a constant target rate, etc.

Competition among fund managers is a well-documented aspect of investment practice for both mutual and hedge funds; see, for example, \cite{chevalier1997risk}, \cite{sirri1998costly}, \cite{brown2001careers}, and \cite{kempf2008tournaments}. These studies typically consider discrete cases involving two fund managers. In contrast, \cite{basak2015competition} proposes a continuous-time log-normal model for two fund managers with power utilities, and \cite{lacker2019mean} examines competition among multiple fund managers with CARA or CRRA utilities in continuous time. Beyond fund managers, the literature also explores competition among insurance companies. For example, \cite{guan2022time} delves into competition among multiple insurers under the mean-variance criterion. Notably, both \cite{lacker2019mean} and \cite{guan2022time} use differences in capital values to directly depict competition, avoiding the use of excess logarithmic returns. Although much of the literature focuses on fund managers’ concern with capital values, studies specifically addressing competition over returns are relatively scarce.

This paper studies the competition among multiple fund managers using logarithmic returns for evaluation and excess logarithmic returns to represent competition. We model two types of fund managers: one optimizing exponential utility for excess logarithmic returns and another using a mean-variance criterion for excess returns. The latter introduces time inconsistency, which we address using the time-consistent Nash equilibrium strategy from \cite{bjork2017time}. Similar to \cite{lacker2019mean} and \cite{guan2022time}, we extend the analysis of multi-player competition from $n$-player games to mean field games (MFGs), as introduced by \cite{lasry2007mean} and \cite{caines2006large}. This paper uses excess logarithmic returns to evaluate funds and examines four models: two $n$-player games and two MFGs, each involving two types of fund managers.

{
Unlike \cite{lacker2019mean} and \cite{guan2022time}, our framework introduces three distinct asset classes available to fund managers: (1) a public risk-free asset, serving as a safe investment with guaranteed returns; (2) a private risk asset, unique to each fund and reflecting idiosyncratic opportunities tied to a manager’s specialization, industry focus, or skill level; and (3) a public risky asset, accessible to all managers and modeled as an index or risk bond. This public risky asset may exhibit positive or negative correlations with private assets, enabling managers to hedge risks or enhance portfolio stability. In practice, managers often allocate to public risky assets for strategic purposes. For example, some fund managers hold industry indices that are negatively correlated with their own performance to hedge risks; others short indices that are positively correlated with their own products to enhance fund stability; and some managers may allocate a portion of their funds to public assets such as bonds due to investment scope or industry restrictions. For simplicity, we assume a single public risky asset, which may correlate variably (positively or negatively) with each manager’s private risky asset. The public risky asset is identical across all managers, serving as a shared tool to navigate systemic risks.  The public risky asset  (e.g., an S\&P 500 index fund or Treasury bond) embodies economy-wide risks (e.g., inflation, rate hikes) that affect all managers simultaneously. For instance, during a recession, multiple funds might short the same equity index to hedge their portfolios, creating collective exposure to its price movements.
}

Our models show that the excess logarithmic return evaluation criterion and the inclusion of public risk assets significantly influence the results. Specifically, the excess logarithmic return evaluation results in a constant equilibrium { allocation fraction}, which differs notably from the findings of \cite{lacker2019mean} and \cite{guan2022time}. Additionally, the presence of public risk assets means that competition among fund managers is primarily reflected in their investments in these assets. In the MFG, competition does not affect investments in private risk assets, whereas in the $n$-player game, such investments are influenced by the number of players 
$n$ and the competition weight parameter for each manager. Fund managers’ decisions are influenced by their own risk aversion and competition weight, but not by the asset parameters of other managers. In the MFG framework, the investment in public risk assets is divided into two components: one that disregards competition and another that addresses competitive risk. The MFE of the MFG serves as the limit of the $n$-player game equilibrium as the number of players approaches infinity. The constant equilibrium investment ratio is practical and is supported by sensitivity analyses performed through derivative calculations and numerical simulations.

The remainder of this paper is organized as follows: Section 2 presents the market model. Section 3 formulates an $n$-player game and an MFG for fund managers with exponential utility for excess logarithmic returns, analyzing and discussing the constant equilibrium and constant MFE for these models. Section 4 extends this to an $n$-player game and an MFG for fund managers using a mean-variance criterion for excess returns, again deriving and analyzing the constant equilibrium and constant MFE. Section 5 concludes the paper, with most proofs provided in the Appendices.

\section{\bf Market model}
Let $\left(\Omega,\mathcal{F},\mathbb{P}\right)$ be a complete probability space with an augmented natural filtration $\left\{\mathcal{F}_t\right\}_{t\geq0}$ generated by $n+1$ independent standard Brownian motions $B$ and $W^k,k=1,2,\cdots,n$. $[0,\mathcal{T}]$ is a fixed investment time horizon.

In the financial market, there are $n$ fund managers and $n$ risky assets $S^k$ for $k=1, 2, \ldots, n$. Each risky asset $S^k$ is exclusively associated with the $k$-th fund manager, representing differences in investment scope, such as stocks, options and futures, virtual currencies, and bonds, among different fund managers. These assets also reflect variations in investment industries and sectors, as well as disparities in the level of fund products due to the diverse styles and strengths of individual fund managers.

Additionally, there is a risk-free bond $S^0$ and a public risk asset $S^*$, which enhance the investment choices available to fund managers and offer opportunities for risk hedging. The public risk asset $S^*$ is not a single stock but is considered an index or a risky bond. In practice, some fund managers hold industry indices negatively correlated with their own track to hedge risks, while others short indices positively correlated with their own products to improve fund stability. Some fund managers may also need to invest a certain { allocation fraction} of their funds in bonds due to restrictions on investment scope and industry, which can lead to weak income potential or large fluctuations over a fixed period. For simplicity, we assume there is only one public risk asset, $S^*$. Thus, the $k$-th fund manager can invest in $S^0$, $S^*$, and $S^k$ satisfying

\begin{eqnarray}
	\frac{\rd S^0_t}{S^0_t} &=& \kappa \rd t,\nonumber\\
	\frac{\rd S^*_t}{S^*_t} &=& \left(\kappa+\mu\right) \rd t + \sigma \rd B_t , \nonumber\\
	\frac{\rd S^k_t}{S^k_t} &=& \left(\kappa+\mu_k\right) \rd t + \sigma_k \rd B_t + \nu_k dW^k_t,\quad k=1,2,\cdots,n \nonumber
\end{eqnarray}
with constant market parameters $\kappa\in\mathbb{R}^{+},\mu\in\mathbb{R}^{+},\mu_k\in\mathbb{R}^{+},\sigma\in\mathbb{R}^{+},\nu_k\in\mathbb{R}^{+}$ and $\sigma_k\in \mathbb{R},$ $ k=1,2,\cdots,n$. The correlation between $S^*$ and $S^k$ is entirely determined by the standard Brownian motion $B$.

{
\begin{remark}
  The risk-free asset $S^0$ has a return rate $\kappa > 0$. Risky assets $S^*$ and $S^k$ are assumed to have higher expected returns than $S^0$ (i.e., $\mu > 0$ and $\mu_k > 0$), making them attractive to risk-averse investors despite their volatility. For simplicity, we take $\sigma > 0$ and $\nu_k > 0$. If $\sigma$ were negative, we could redefine the Brownian motion $B$ as $-B$ to make $\sigma$ positive. Similarly, $\nu_k$ represents idiosyncratic volatility and is inherently non-negative. The parameter $\sigma_k$, however, can take negative values. This is because the correlation between $S^*$ and $S^k$ is given by
  \[
    \rho_k = \frac{\sigma_k}{\sqrt{\sigma_k^2 + \nu_k^2}},
  \]
  where the sign of $\sigma_k$ determines whether $S^*$ and $S^k$ are positively or negatively correlated. Allowing $\sigma_k \in \mathbb{R}$ enables us to model both positive and negative correlations.
\end{remark}
}

\vspace{0.5cm}

{ Let $X^k_t$ denote the total wealth of the $k$-th fund manager at time $t$, and denote $\alpha^k_t$ and $\beta^k_t$ as the allocation fractions of the $k$-th fund manager's investment in $S^*$ and $S^k$ at time $t$, respectively. Since short-selling is permitted in the financial market, the values of $\alpha^k_t$ and $\beta^k_t$ belong to the set of real numbers $\mathbb{R}$. Define $\pi^k_t := \begin{pmatrix}
	\alpha^k_t \\ \beta^k_t
\end{pmatrix}$ , $t\leq \mathcal{T}$,
and $\pi^k := \{\pi^k_t\}_{t\in\left[0,\mathcal{T}\right]}$ is the strategy of the $k$-th fund manager.  The wealth process of the $k$-th fund manager  $X^{k,\pi^k}=\left\{X^{k,\pi^k}_t\right\}_{t\in\left[0,\mathcal{T}\right]}$ with strategy $\pi^k$ evolves as follows: }
\begin{eqnarray}
	\rd X^{k,\pi^k}_t &=& X^{k,\pi^k}_t\left[\left(\kappa + \mu \alpha^k_t + \mu_k\beta^k_t \right) \rd t + \left(\sigma \alpha^k_t + \sigma_k\beta^k_t \right) \rd B_t + \nu_k\beta^k_t \rd W^k_t\right],\nonumber\\
	X^{k,\pi^k}_0 &=& x^{k,\pi^k}_0, \quad k=1,2,\cdots,n. \nonumber
\end{eqnarray}
We consider the logarithmic return of the $k$-th fund manager $R^{k,\pi^k}_t=\log\left(\frac{X^{k,\pi^k}_t}{x^{k,\pi^k}_0}\right)$ evolving as
\begin{eqnarray}
	\rd R^{k,\pi^k}_t &=& \left[\left(\kappa + \mu \alpha^k_t + \mu_k\beta^k_t \right) - \frac{1}{2}\left(\sigma \alpha^k_t + \sigma_k\beta^k_t\right)^2 - \frac{1}{2}\left(\nu_k\beta^k_t\right)^2\right] \rd t \nonumber \\&&+ \left(\sigma \alpha^k_t + \sigma_k\beta^k_t \right) \rd B_t + \nu_k\beta^k_t \rd W^k_t,\nonumber\\
	R^{k,\pi^k}_0 &=& 0, \quad k=1,2,\cdots,n. \nonumber
\end{eqnarray}

{
For the \( k \)-th fund manager, the utility function \( U^k(\cdot) \) is based on the excess logarithmic return at the terminal time \( \mathcal{T} \), given by \( R^{k, \pi^k}_\mathcal{T} - \vartheta \), where \( \vartheta \) represents a benchmark logarithmic return. The benchmark \( \vartheta \) is defined as the market average return, which is a proportion \( \theta_k \) of the average log-returns of the \( n \) fund managers at terminal time \( \mathcal{T} \). Specifically, 
\[
\vartheta = \theta_k \overline{R_\mathcal{T}}, \quad \overline{R_\mathcal{T}} := \frac{1}{n} \sum_{i=1}^{n} R^{i, \pi^i}_\mathcal{T}, \quad 0 \le \theta_k \le 1.
\]
}

Additionally, we assert that different fund managers have varying criteria and attitudes toward the benchmark. Specifically, $\theta_k$ is referred to as the competition weight parameter for the $k$-th fund manager. When $\theta_k=0$,  the manager is not concerned with the performance of other managers. The managers consider utility functions that fall into two main categories: expected utility criteria and mean-variance criteria. The framework for fund manager competition is built around these criteria, as studied in Section \ref{Exponential utility} and Section \ref{Mean-Variance criterion}.

\newpage
\section{{{\bf Exponential utility criterion}}}\label{Exponential utility}

{
We assume that the fund managers' risk aversion is independent of the excess logarithmic return \( R^{k, \pi^k}_\mathcal{T} - \theta_k \frac{1}{n} \sum_{i=1}^{n} R^{i, \pi^i}_\mathcal{T} \) and remains constant. Specifically, the risk aversion is given by \( -\frac{U^{k}_{zz}}{U^{k}_{z}} \equiv \frac{1}{\delta_k} \) for \( k = 1, 2, \cdots, n \),} where $\delta_k>0$ is the absolute risk aversion coefficient. Then the utility function of the $k$-th fund manager $U^k$ is an exponential utility function given by $$U^k(z)=-\exp\left(-\frac{1}{\delta_k}z\right)$$ and the payoff of the $k$-th fund manager is given by { 
\begin{eqnarray}
	J_k\left(\pi^k;\pi^{-k}\right):=\mathbb{E}\left\{-\exp\left[-\frac{1}{\delta_k}\left(R^{k,\pi^k}_\mathcal{T}-\frac{\theta_k}{n}\sum_{i=1}^{n}R^{i,\pi^i}_\mathcal{T}\right)\right]\right\},\nonumber
\end{eqnarray}}
where $\pi^{-k} = \left(\pi^1,\cdots,\pi^{k-1},\pi^{k+1},\cdots,\pi^n\right)$ represents the strategies of the other $n-1$ fund managers and $\mathbb{E}$ represents the expectation. $\delta_k>0$ represents the $k$-th fund manager's risk aversion coefficient. A higher value of $\delta_k$ indicates that the manager is more risk-averse.

\begin{remark}
	While both the aforementioned utility function and the CARA utility are exponential in nature and assume constant absolute risk aversion, they differ fundamentally. The CARA utility is typically employed with absolute capital or consumption, whereas in this paper, it is applied to excess logarithmic returns.
\end{remark}
\begin{definition}\label{def1}
	A strategy $\left(\pi^{1},\cdots,\pi^{n}\right)$ is admissible if, for any $k=1,\cdots,n$, $\pi^k$ is  $\left\{\mathcal{F}_t\right\}_{t\geq0}$- progressively measurable and satisfies { 
	\begin{eqnarray}
		\mathbb{E}\int_{0}^{\mathcal{T}}\left[(\alpha^{k}_t)^2+(\beta^{k}_t)^2\right] \rd t < +\infty.\nonumber
	\end{eqnarray}}
An admissible strategy $\left(\pi^{1,*},\cdots,\pi^{n,*}\right)$ constitutes an equilibrium if, for any $\pi^k$ and $k=1,\cdots,n$,
	\begin{eqnarray}
		J_k\left(\pi^{k,*};\pi^{-k,*}\right)\ge J_k\left(\pi^{k};\pi^{-k,*}\right).\nonumber
	\end{eqnarray}
 A constant equilibrium is an equilibrium in which, for each $k$, both $\alpha^{k,*}$ and $\beta^{k,*}$ are constant over time.
\end{definition}

If the public risky asset $S^*$ is eliminated, our model aligns mathematically with the one outlined in Section 3 of \cite{lacker2019mean}. Nevertheless, the two models diverge in an economic context. Our approach initiates from the rate of return, assuming absolute risk aversion that is unrelated to returns. Conversely, \cite{lacker2019mean} assumes relative risk aversion tied to wealth. Additionally, the results differ significantly when public risk assets are considered.
\begin{theorem}\label{the1}
	There exists a unique constant equilibrium, given by
	\begin{eqnarray}
		\alpha^{k,*}&=&\frac{\delta_k\sigma_k\left(\mu\sigma_k-\mu_k\sigma\right)}{\left(1+\delta_k-\frac{\theta_k}{n}\right)\sigma^2\nu_k^2}+\frac{\mu\delta_k+\theta_k\sigma D}{\left(1+\delta_k\right)\sigma^2},\quad k=1,\cdots,n,\label{alp1}\\
		\beta^{k,*}&=&\frac{\delta_k\left(\mu_k\sigma-\mu\sigma_k\right)}{\left(1+\delta_k-\frac{\theta_k}{n}\right)\sigma\nu_k^2},\quad k=1,\cdots,n,\label{bet1}
	\end{eqnarray}
	where
	\begin{eqnarray}
		D:= \frac{\frac{1}{n}\sum_{i=1}^{n}\frac{\delta_k}{1+\delta_k}\frac{\mu}{\sigma}}{1-\frac{1}{n}\sum_{i=1}^{n}\frac{\theta_k}{1+\delta_k}}.\label{D1}
	\end{eqnarray}
\end{theorem}

\begin{proof}
	See Appendix \ref{proof1}.
\end{proof}

\subsection{{{\bf Analysis about the constant equilibrium $\left(\alpha^{k,*},\beta^{k,*}\right)$}}}
\label{analysis1}

For the $k$-th fund manager, the { allocation fractions} of investment in $S^*$ and $S^k$, denoted by $\alpha^{k,*}$ and $\beta^{k,*}$, are both constant. The { allocation fraction} $\alpha^{k,*}$ is influenced by the parameters of other investment managers, including their competition weight parameters $\theta_i$ (for $i \neq k$) and risk aversion coefficients $\delta_i$ (for $i \neq k$), but does not depend on the parameters related to the risky assets $S^i$ (for $i \neq k$). In contrast, $\beta^{k,*}$ is not affected by the parameters of other investment managers but is influenced by  the own competition weight parameter $\theta_k$.

{
The structure of this subsection is as follows. First, we analyze the investment ratio of the \( k \)-th fund manager in \( S^k \), denoted as \( \beta^{k,*} \). Next, we examine the investment ratio of the \( k \)-th fund manager in \( S^* \), which is denoted as \( \alpha^{k,*} \). A detailed analysis is provided in Remark \ref{remark_new}. We define the virtual Sharpe ratio for \( S^k \) as \( \frac{\mu_k}{\sigma_k} \) and consider three cases: (Case 1) when \( \sigma_k > 0 \) and the Sharpe ratio of the public asset \( S^* \) exceeds that of \( S^k \), i.e., \( \sigma_k > 0 \) and \( \frac{\mu}{\sigma} > \frac{\mu_k}{\sigma_k} \); (Case 2) when \( \sigma_k > 0 \) and the virtual Sharpe ratio of \( S^k \) exceeds that of \( S^* \), i.e., \( \sigma_k > 0 \) and \( \frac{\mu_k}{\sigma_k} > \frac{\mu}{\sigma} \); and (Case 3) when \( \sigma_k < 0 \), indicating that \( S^* \) and \( S^k \) exhibit opposite fluctuations.
}

{
\begin{remark}
The Sharpe ratio of asset \( S^* \) is given by \( \frac{\mu}{\sigma} \), and the Sharpe ratio of asset \( S^k \) is \( \frac{\mu_k}{\sqrt{\sigma_k^2 + \nu_k^2}} \). When the influence of the Brownian motion \( W^k \), which is independent of the Brownian motion \( B \), is ignored (i.e., when \( \nu_k = 0 \)), the Sharpe ratio of \( S^k \) simplifies to \( \frac{\mu_k}{|\sigma_k|} \). At this stage, the correlation between \( S^* \) and \( S^k \) (i.e., the sign of \( \sigma_k \)) has not yet been considered. In the subsequent analyses, it is often necessary to compare the Sharpe ratio of \( S^* \), \( \frac{\mu}{\sigma} \), with \( \frac{\mu_k}{\sigma_k} \). The latter can be seen as a Sharpe-like ratio that accounts for the positive and negative correlations when only \( B \) is considered and \( W^k \) is excluded. This ratio is referred to as the virtual Sharpe ratio.
	\end{remark}
}

{
\subsubsection{{\bf Analysis of \( \beta^{k,*} \)}}

Let us begin by analyzing \( \beta^{k,*} \), the { allocation fraction} of investment in \( S^k \). First, we examine the impact of parameters related to competition, including the parameters of other fund managers, the number of fund managers \( n \), and the individual competition weight parameter \( \theta_k \). Next, we consider the parameters of the risky assets (\( \mu, \sigma, \mu_k, \sigma_k, \nu_k \)) and the fund manager’s own risk aversion coefficient \( \delta_k \).

Although \( \beta^{k,*} \) is independent of the parameters of other fund managers, it does not represent the optimal solution to the classical individual optimal investment problem described by Eq.~(\ref{bet1}) with \( \theta_k = 0 \). Rather, it is a modification of this classical solution that incorporates competition effects.} One obtains
\begin{eqnarray}
	\frac{\partial\beta^{k,*}}{\partial\theta_k}=\frac{\delta_k\left(\mu_k\sigma-\mu\sigma_k\right)}{n\left(1+\delta_k-\frac{\theta_k}{n}\right)^2\sigma\nu_k^2}.\label{partial1}
\end{eqnarray}
For case 1, $\beta^{k,*} < 0$ and $\frac{\partial \beta^{k,*}}{\partial \theta_k} < 0$, indicating that if the Sharpe ratio of $S^*$ is higher than the virtual Sharpe ratio of $S^k$, the fund manager will short $S^k$, with the { allocation fraction} of short sales increasing as $\theta_k$ rises. For cases 2 and 3, $\beta^{k,*} > 0$ and $\frac{\partial \beta^{k,*}}{\partial \theta_k} > 0$, implying that if the virtual Sharpe ratio of $S^k$ exceeds the Sharpe ratio of $S^*$ or $\sigma_k < 0$, the fund manager will take a long position in $S^k$, with the { allocation fraction} of the long position increasing as $\theta_k$ grows. Regardless of the case, as more fund managers participate in the competition, { $\beta^{k,*}$} approaches the optimal solution in the classical framework, given by
\begin{eqnarray}
	\beta^{k,*}_{\infty}:=\lim\limits_{n\rightarrow +\infty}\beta^{k,*} = \frac{\delta_k\left(\mu_k\sigma-\mu\sigma_k\right)}{\left(1+\delta_k\right)\sigma\nu_k^2}.\label{limit1}
\end{eqnarray}

{ For the parameters of the risky assets and the individual's own risk aversion coefficient,} it is straightforward to verify that
\begin{eqnarray}
	&&\frac{\partial\beta^{k,*}}{\partial\mu}=\frac{-\delta_k\sigma_k}{\left(1+\delta_k-\frac{\theta_k}{n}\right)\sigma\nu_k^2},\qquad
	\frac{\partial\beta^{k,*}}{\partial\sigma}=\frac{\delta_k\mu\sigma_k}{\left(1+\delta_k-\frac{\theta_k}{n}\right)\sigma^2\nu_k^2},\nonumber\\
	&&\frac{\partial\beta^{k,*}}{\partial\mu_k}=\frac{\delta_k}{\left(1+\delta_k-\frac{\theta_k}{n}\right)\nu_k^2},\qquad\;\;\;
	\frac{\partial\beta^{k,*}}{\partial\sigma_k}=\frac{-\delta_k\mu}{\left(1+\delta_k-\frac{\theta_k}{n}\right)\sigma\nu_k^2},\nonumber\\
	&&\frac{\partial\beta^{k,*}}{\partial\nu_k}=\frac{-2\delta_k\left(\mu_k\sigma-\mu\sigma_k\right)}{\left(1+\delta_k-\frac{\theta_k}{n}\right)\sigma\nu_k^3},\qquad
	\frac{\partial\beta^{k,*}}{\partial\delta_k}=\frac{\left(1-\frac{\theta_k}{n}\right)\left(\mu_k\sigma-\mu\sigma_k\right)}{\left(1+\delta_k-\frac{\theta_k}{n}\right)^2\sigma\nu_k^2}.\label{partial2}
\end{eqnarray}

For case 1, the sensitivity of $\beta^{k,*}$ to various parameters is as follows:, $\frac{\partial\beta^{k,*}}{\partial\mu}<0,\frac{\partial\beta^{k,*}}{\partial\sigma}>0,\frac{\partial\beta^{k,*}}{\partial\mu_k}>0,\frac{\partial\beta^{k,*}}{\partial\sigma_k}<0,\frac{\partial\beta^{k,*}}{\partial\nu_k}>0$ and $\frac{\partial\beta^{k,*}}{\partial\delta_k}<0$. Thus, $\beta^{k,*}$ increases as $\mu$ decreases, as $\sigma$ increases, as $\mu_k$ increases, as $\sigma_k$ decreases, as $\nu_k$ increases, and as $\delta_k$ decreases.

For case 2, the sensitivities of $\beta^{k,*}$ to various parameters are: $\frac{\partial\beta^{k,*}}{\partial\mu}<0,\frac{\partial\beta^{k,*}}{\partial\sigma}>0,\frac{\partial\beta^{k,*}}{\partial\mu_k}>0,\frac{\partial\beta^{k,*}}{\partial\sigma_k}<0,\frac{\partial\beta^{k,*}}{\partial\nu_k}<0$ and $\frac{\partial\beta^{k,*}}{\partial\delta_k}>0$. Thus, $\beta^{k,*}$ increases as $\mu$ decreases, as $\sigma$ increases, as $\mu_k$ increases, as $\sigma_k$ decreases, as $\nu_k$ decreases, and as $\delta_k$ increases.

For case 3, the sensitivities of $\beta^{k,*}$ to various parameters are: $\frac{\partial\beta^{k,*}}{\partial\mu}>0,\frac{\partial\beta^{k,*}}{\partial\sigma}<0,\frac{\partial\beta^{k,*}}{\partial\mu_k}>0,\frac{\partial\beta^{k,*}}{\partial\sigma_k}<0,\frac{\partial\beta^{k,*}}{\partial\nu_k}<0$ and $\frac{\partial\beta^{k,*}}{\partial\delta_k}>0$. $\beta^{k,*}$ increases as $\mu$ increases, increases as $\sigma$ decreases, increases as $\mu_k$ increases, increases as $\sigma_k$ decreases, increases as $\nu_k$ decreases, and increases as $\delta_k$ increases.

In general, the greater the excess mean rate of return $\mu_k$ of $S^k$ or the smaller the volatility coefficient $\sigma_k$ of $S^k$, the larger the { allocation fraction} of investment in $S^k$. When $S^*$ and $S^k$ are negatively correlated, the $k$-th fund manager must take a long position in $S^k$, with the { allocation fraction} increasing as the excess mean rate of return $\mu$ of $S^*$ rises or the volatility coefficient $\sigma$ of $S^*$ decreases. Conversely, when $S^*$ and $S^k$ are positively correlated, the $k$-th fund manager will short $S^k$ if the virtual Sharpe ratio is low and go long if it is high. The effects of $\mu$ and $\sigma$ of $S^*$ on the { allocation fraction} in $S^k$ are opposite to those in the case of negative correlation between $S^*$ and $S^k$. Additionally, the second volatility coefficient $\nu_k$ of $S^k$ and the risk aversion coefficient $\delta_k$ do not determine whether the { allocation fraction} of investment in $S^k$ is positive or negative but influence its absolute value. Specifically, a smaller $\nu_k$ or a larger $\delta_k$ leads to a larger { allocation fraction} of short or long positions in $S^k$.

{ 

}

{
\subsubsection{{\bf Analysis of \( \alpha^{k,*} \)}}
Let us analyze the { allocation fraction} of investment in \( S^* \) for the \( k \)-th fund manager, denoted as \( \alpha^{k,*} \). First, we examine the parameters of other fund managers, including their competition weight parameters and risk aversion coefficients. Next, we consider the parameters related to the individual fund manager, such as the parameters associated with the public risky asset \( S^* \) and the personal risky asset \( S^k \), as well as the individual's competition weight parameter \( \theta_k \) and risk aversion coefficient \( \delta_k \). The analysis of the risk aversion coefficient requires numerical plotting. Finally, we explore the scenario of an infinite number of fund managers, which leads to the mean-field game (MFG) framework discussed in the next subsection.
}

The { allocation fraction} $\alpha^{k,*}$ is influenced not only by the fund manager's own parameters but also by the parameters of other fund managers. Specifically, for any $i \neq k$,
\begin{eqnarray}
	\frac{\partial\alpha^{k,*}}{\partial\theta_i}&=&\frac{\theta_k\mu}{\left(1+\delta_k\right)\sigma^2}\frac{\frac{1}{n}\sum_{j=1}^{n}\frac{\delta_j}{1+\delta_j}}{\left(1-\frac{1}{n}\sum_{j=1}^{n}\frac{\theta_j}{1+\delta_j}\right)^2}\frac{1}{n\left(1+\delta_i\right)}>0,\label{partial3}\\
	\frac{\partial\alpha^{k,*}}{\partial\delta_i}&=&\frac{\theta_k\mu}{\left(1+\delta_k\right)\sigma^2}\frac{1-\frac{1}{n}\sum_{j=1}^{n}\frac{\theta_j}{1+\delta_j}-\frac{1}{n}\sum_{j=1}^{n}\frac{\delta_j}{1+\delta_j}\theta_i}{n\left(1+\delta_i\right)^2\left(1-\frac{1}{n}\sum_{j=1}^{n}\frac{\theta_j}{1+\delta_j}\right)^2}\nonumber\\
	&\ge& \frac{\theta_k\mu}{\left(1+\delta_k\right)\sigma^2}\frac{1-\max_{j}\theta_j}{n\left(1+\delta_i\right)^2\left(1-\frac{1}{n}\sum_{j=1}^{n}\frac{\theta_j}{1+\delta_j}\right)^2} >0.\label{partial4}
\end{eqnarray} 
Any increase in the competition weight parameter $\theta_i$ or the risk aversion coefficient $\delta_i$ of any other fund manager $i$ will lead to an increase in the { allocation fraction} of the $k$-th fund manager's investment in the public risk asset $S^*$.

The effect of the parameters related to the $k$-th fund manager on $\alpha^{k,*}$ is complex, particularly in case 2. Therefore, we will analyze only case 1, where $\sigma_k > 0$ and $\frac{\mu}{\sigma} > \frac{\mu_k}{\sigma_k}$, and case 3, where $\sigma_k < 0$.

\begin{eqnarray}
	\frac{\partial\alpha^{k,*}}{\partial\mu}&=&\frac{\delta_k\sigma_k^2}{\left(1+\delta_k-\frac{\theta_k}{n}\right)\sigma^2\nu_k^2}+\frac{\delta_k}{\left(1+\delta_k\right)\sigma^2}+\frac{\theta_k}{\left(1+\delta_k\right)\sigma^2}\frac{\frac{1}{n}\sum_{j=1}^{n}\frac{\delta_j}{1+\delta_j}}{1-\frac{1}{n}\sum_{j=1}^{n}\frac{\theta_j}{1+\delta_j}}>0,\nonumber\\
	\frac{\partial\alpha^{k,*}}{\partial\sigma}&=&\frac{-2\delta_k\mu\sigma_k^2+\delta_k\mu_k\sigma_k\sigma}{\left(1+\delta_k-\frac{\theta_k}{n}\right)\sigma^3\nu_k^2}-\frac{2\delta_k\mu}{\left(1+\delta_k\right)\sigma^3}-\frac{2\mu\theta_k}{\left(1+\delta_k\right)\sigma^3}\frac{\frac{1}{n}\sum_{j=1}^{n}\frac{\delta_j}{1+\delta_j}}{1-\frac{1}{n}\sum_{j=1}^{n}\frac{\theta_j}{1+\delta_j}}<0,\nonumber\\
	\frac{\partial\alpha^{k,*}}{\partial\mu_k}&=&\frac{-\delta_k\sigma_k}{\left(1+\delta_k-\frac{\theta_k}{n}\right)\sigma\nu_k^2}\left\{
	\begin{aligned}
		<0\quad case \quad 1\\
		>0\quad case \quad 3
	\end{aligned}
	\right.,\nonumber\\
	\frac{\partial\alpha^{k,*}}{\partial\sigma_k}&=&\frac{\delta_k\left(2\mu\sigma_k-\mu_k\sigma\right)}{\left(1+\delta_k-\frac{\theta_k}{n}\right)\sigma^2\nu_k^2}\left\{
	\begin{aligned}
		>0\quad case\quad  1\\
		<0\quad case \quad 3
	\end{aligned}
	\right.,\nonumber\\
	\frac{\partial\alpha^{k,*}}{\partial\nu_k}&=&\frac{-2\delta_k\sigma_k\left(\mu\sigma_k-\mu_k\sigma\right)}{\left(1+\delta_k-\frac{\theta_k}{n}\right)\sigma^2\nu_k^3}<0,\nonumber\\
	\frac{\partial\alpha^{k,*}}{\partial\theta_k}&=&\frac{\delta_k\sigma_k\left(\mu\sigma_k-\mu_k\sigma\right)}{n\left(1+\delta_k-\frac{\theta_k}{n}\right)^2\sigma^2\nu_k^2}+\frac{\mu}{\left(1+\delta_k\right)\sigma^2}\frac{\frac{1}{n}\sum_{j=1}^{n}\frac{\delta_j}{1+\delta_j}}{1-\frac{1}{n}\sum_{j=1}^{n}\frac{\theta_j}{1+\delta_j}},\nonumber\\
	&&+\frac{\mu\theta_k}{n\left(1+\delta_k\right)^2\sigma^2}\frac{\frac{1}{n}\sum_{j=1}^{n}\frac{\delta_j}{1+\delta_j}}{\left(1-\frac{1}{n}\sum_{j=1}^{n}\frac{\theta_j}{1+\delta_j}\right)^2}>0.\label{partial5}
\end{eqnarray}

Regardless of whether it is case 1 or case 3, the { allocation fraction} $\alpha^{k,*}$ of the investment in $S^*$ increases as the excess mean rate of return $\mu$ of $S^*$ increases, as the volatility coefficient $\sigma$ of $S^*$ decreases, as the second volatility coefficient $\nu_k$ of $S^k$ decreases, and as the competition weight parameter $\theta_k$ of the $k$-th fund manager increases. In case 1, $\alpha^{k,*}$ also increases as the excess mean rate of return $\mu_k$ of $S^k$ decreases and as the first volatility coefficient $\sigma_k$ of $S^k$ increases. In contrast, in case 3, the effects are reversed.

Further, for analyzing the impact of $\delta_k$ in cases 1 and 3, and the effects of both $\delta_k$ and $\theta_k$ in case 2, we consider two special scenarios: one with only two fund managers and the other with an infinite number of fund managers.

In the case of two fund managers, the analysis needs to be conducted numerically. We take $n=2,\left(\mu,\sigma\right)=\left(1,1\right),\theta_2=0.5,\delta_2=5,\delta_1\in\left(0,20\right)$, $\theta_1=0.5,\left(\mu_1,\sigma_1,\nu_1\right)=\left(2,3,3\right)$ in case 1, $\theta_1=0.5,\left(\mu_1,\sigma_1,\nu_1\right)=\left(2,-1,1\right)$ in case 3 and $\theta_1\in\left(0,1\right),\left(\mu_1,\sigma_1,\nu_1\right)=\left(3,2,2\right)$ in case 2. Figure~\ref{Third1and2} illustrates that $\alpha^{1,*}$ increases as the risk aversion coefficient $\delta_1$ increases, regardless of whether it is case 1 or case 3. However, in case 2, Figure~\ref{Third3} shows that $\alpha^{1,*}$ increases with $\delta_1$ when $\theta_1$ is small, while $\alpha^{1,*}$ decreases with $\delta_1$ when $\theta_1$ is large. Additionally, $\alpha^{1,*}$ increases as the competition weight parameter $\theta_1$ increases.

\begin{figure}[H]
	\centering
	\includegraphics[scale=0.71]{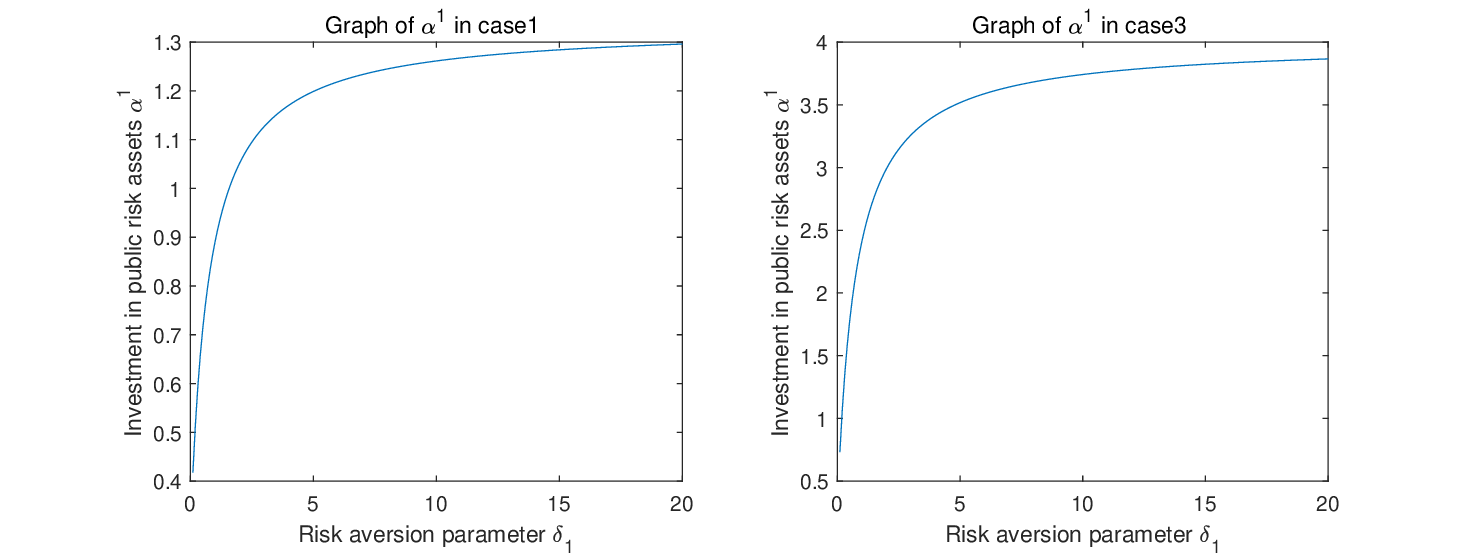}
	\caption{Investment in public risk assets $\alpha^{1,*}$ in case 1 and case 3}
	\label{Third1and2}
\end{figure} 

\begin{figure}[H]
	\centering
	\includegraphics[scale=0.89]{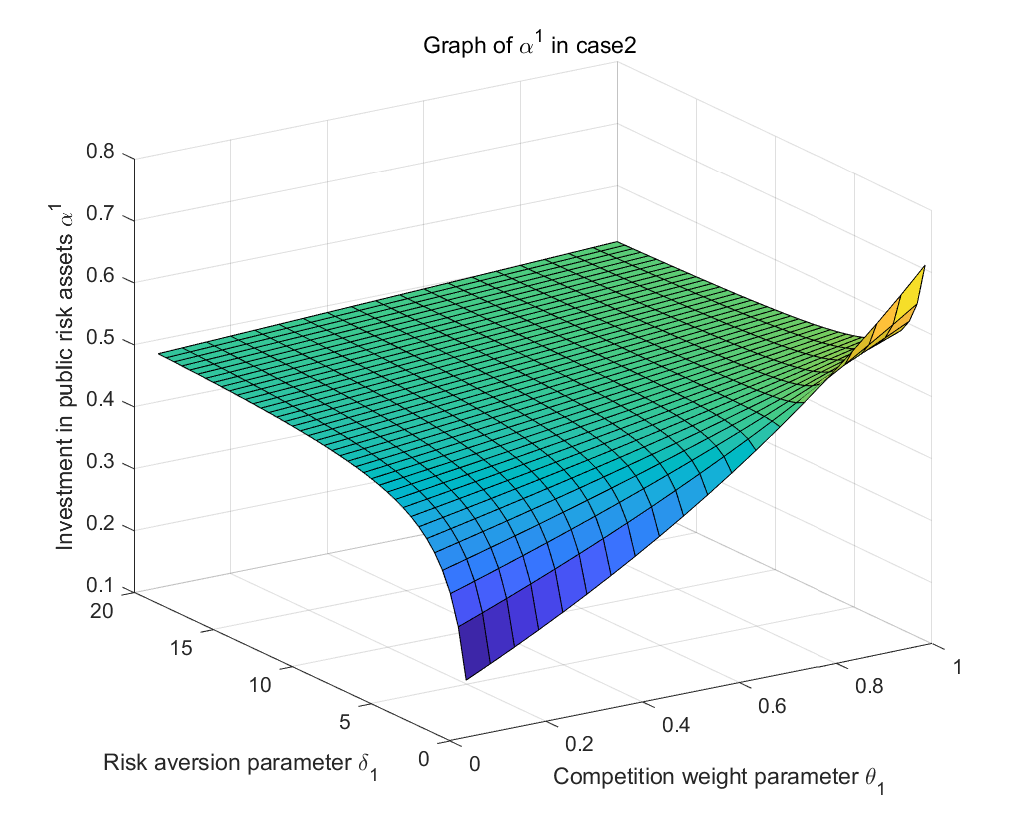}
	\caption{Investment in public risk assets $\alpha^{1,*}$ in case 2}
	\label{Third3}
\end{figure}

{
\subsubsection{{\bf The case of an infinite number of fund managers}}\label{infity_case}

Now, we consider the case of an infinite number of fund managers, i.e., \( n \to +\infty \). This is an informal argument designed to build intuition and motivate the upcoming definition. For the \( k \)-th fund manager, we define the type vector
\[
\phi_k = \left( \delta_k, \theta_k, \mu_k, \sigma_k, \nu_k \right),
\]
which induces an empirical measure known as the type distribution. This type distribution is a probability measure on the type space
\[
\mathcal{Z} = (0, \infty) \times [0, 1] \times (0, \infty) \times \mathbb{R} \times (0, \infty),
\]
given by
\[
m_n(A) = \frac{1}{n} \sum_{i=1}^{n} 1_A(\phi_k), \quad \text{for Borel sets} \, A \subset \mathcal{Z}.
\]
We then observe that the constants \( D_n \) (Eq.~(\ref{D1})) are obtained by integrating appropriate functions under \( m_n \).

Next, assume that as the number of fund managers increases and \( n \to +\infty \), the empirical measure \( m_n \) converges weakly to a measure \( m \), in the sense that
\[
\int_{\mathcal{Z}} f \, \mathrm{d}m_n \to \int_{\mathcal{Z}} f \, \mathrm{d}m \quad \text{for every bounded continuous function} \, f \, \text{on} \, \mathcal{Z}.
\]
For instance, if the \( \phi_k \)'s are independent and identically distributed (i.i.d.) samples from the distribution \( m \), this convergence holds almost surely. Let \( \phi = (\delta, \theta, \mu, \sigma, \nu) \) be a random variable following the limiting distribution \( m \). In this case, we expect \( D_n \) (as defined in Eq.~(\ref{D1})) to converge to
\[
\lim_{n \to +\infty} D_n = \frac{\mathbb{E} \frac{\delta}{1+\delta}}{1 - \mathbb{E} \frac{\theta}{1+\delta}} \frac{\mu}{\sigma} := L \frac{\mu}{\sigma},
\]
where
\begin{equation}
L := \frac{\mathbb{E} \frac{\delta}{1+\delta}}{1 - \mathbb{E} \frac{\theta}{1+\delta}}. \label{L}
\end{equation}
}


Therefore, when a sufficiently large number of fund managers are participating in the competition, the { allocation fraction} of the $k$-th fund manager in $S^*$ converges to a specific value.
\begin{eqnarray}
	\alpha^{k,*}_{\infty}:=\lim\limits_{n\rightarrow +\infty}\alpha^{{k,*}}=\frac{\delta_k\sigma_k\left(\mu\sigma_k-\mu_k\sigma\right)}{\left(1+\delta_k\right)\sigma^2\nu_k^2}+\frac{\mu\delta_k+\theta_k\mu L}{\left(1+\delta_k\right)\sigma^2}.\label{limit2}
\end{eqnarray}

By taking the derivative and rearranging the terms,
\begin{eqnarray}
	\frac{\partial\alpha^{k,*}_{\infty}}{\partial\theta_k}&=&\frac{\mu L}{\left(1+\delta_k\right)\sigma^2}>0,\nonumber\\
	\frac{\partial\alpha^{k,*}_{\infty}}{\partial\delta_k}&=& \frac{1}{\left(1+\delta_k\right)^2}\left[\frac{\sigma_k\left(\mu\sigma_k-\mu_k\sigma\right)}{\sigma^2\nu_k^2} + \frac{\mu\left(1-\theta_k L\right)}{\sigma^2}\right]\nonumber\\
	&\ge& \frac{1}{\left(1+\delta_k\right)^2}\left[\frac{\sigma_k\left(\mu\sigma_k-\mu_k\sigma\right)}{\sigma^2\nu_k^2} + \frac{\mu\left(1-\theta_k\right)\left(1-\mathbb{E}\frac{1}{1+\delta}\right)}{\sigma^2\left(1-\mathbb{E}\frac{\theta}{1+\delta}\right)}\right]>0.\nonumber
\end{eqnarray}
The analysis of $\theta_k$ and $\delta_k$ is the same as the analysis in the case of two fund managers.

{ 
\begin{remark}
        The MFG framework, which is defined later, provides a structure for deriving the limiting strategy as the outcome of a self-contained equilibrium problem. Instead of directly modeling a continuum of fund managers, we model a single representative fund manager, regarded as randomly selected from the population. Intuitively, this setup represents a game involving a continuum of fund managers with a type distribution $m$.
\end{remark}
}

{ 
	\begin{remark}\label{remark_new}

Let’s consider \( \alpha^{k,*} \) and \( \beta^{k,*} \) jointly. When \( S^* \) and \( S^k \) are negatively correlated, both assets are long-bought, indicating a hedging and complementary relationship. An increase in the quality of either asset (i.e., higher Sharpe or virtual Sharpe ratio) results in a higher { allocation fraction} in that asset.

When \( S^* \) and \( S^k \) are positively correlated, they are in a competitive relationship. As the quality of one asset improves, its  { allocation fraction} increases while the other’s generally decreases. Specifically, if the virtual Sharpe ratio of \( S^k \) is lower than the Sharpe ratio of \( S^* \), \( S^* \) is long-bought and \( S^k \) is short-sold. If the virtual Sharpe ratio of \( S^k \) exceeds that of \( S^* \), \( S^k \) is long-bought, and \( S^* \) can either be long-bought or short-sold.
	\end{remark}
}

\subsection{{{\bf The MFG under exponential utility criterion}}}\label{Themfg}
In this subsection, we revise the financial model to construct the Mean Field Game (MFG) under the exponential utility criterion, assuming an infinite number of fund managers. In this case, the parameters of the fund managers are treated as random variables. To formulate the MFG, we now assume that the probability space $(\Omega, \mathcal{F}, \mathbb{P})$ supports an additional independent one-dimensional Brownian motion $W$, as well as a random variable (referred to as the type vector)
\begin{eqnarray}
	\phi = \left(\delta,\theta,\widehat{\mu},\widehat{\sigma},\widehat{\nu}\right)\nonumber
\end{eqnarray}
with the values in the space $\left(0,+\infty\right)\times\left[0,1\right)\times\left(0,+\infty\right)\times\mathbb{R}\times\left(0,+\infty\right)$, which is independent of $W$ and $B$. The distribution of this random variable $\phi$ is called the type distribution. In fact, $\left(\delta,\theta,\widehat{\mu},\widehat{\sigma},\widehat{\nu}\right)$ is exactly the randomization of the parameters $\left(\delta_k,\theta_k,\mu_k,\sigma_k,\nu_k\right)$ of a particular fund manager. The fund manager with type vector $\phi$ can invest in $S^0$, $S^*$, and its exclusive asset $\widehat{S}$ satisfying
\begin{eqnarray}
	\frac{\rd \widehat{S}_t}{\widehat{S}_t} &=& \left(\kappa+\widehat{\mu}\right) \rd t + \widehat{\sigma} \rd B_t + \widehat{\nu} dW_t. \nonumber
\end{eqnarray}

{ Let \( \mathcal{F}^{MF} = \left\{ \mathcal{F}^{MF}_t \right\}_{t \in [0, \mathcal{T}]} \) denote the smallest filtration satisfying the usual conditions, such that \( \phi \) is \( \mathcal{F}^{MF}_0 \)-measurable and both \( W \) and \( B \) are adapted. Additionally, let \( \mathcal{F}^{B} = \left\{ \mathcal{F}^{B}_t \right\}_{t \in [0, \mathcal{T}]} \) represent the natural filtration generated by \( B \).}

{ The logarithmic return $\left\{R_t^\pi\right\}_{t\in \left[0,\mathcal{T}\right]}$ of a fund manager with type vector $\phi$ evolves as }
\begin{eqnarray}
	\rd R_t^\pi &=& \left[\left(\kappa + \mu \alpha_t + \widehat{\mu}\beta_t \right) - \frac{1}{2}\left(\sigma \alpha_t + \widehat{\sigma}\beta_t\right)^2 - \frac{1}{2}\left(\widehat{\nu}\beta_t\right)^2\right] \rd t \nonumber \\&&+ \left(\sigma \alpha_t + \widehat{\sigma}\beta_t \right) \rd B_t + \widehat{\nu}\beta_t \rd W_t,\label{Rt}\\
	R_0^\pi &=& 0.\nonumber
\end{eqnarray}
A strategy $\pi=\left(\alpha,\beta\right)$ is admissible if it is $\mathcal{F}^{MF}$ progressively measurable and satisfies {
\begin{eqnarray}
	\mathbb{E}\int_{0}^{\mathcal{T}}\left(\alpha^{2}_t+\beta^{2}_t\right) \rd t < \infty.\label{admissible}
\end{eqnarray} 
For an admissible strategy \( \eta \), define the \( \mathcal{F}^{B}_\mathcal{T} \)-measurable random process \( \left\{ \overline{R}_t = \mathbb{E}\left[ R^{\eta}_t \mid \mathcal{F}^{B}_\mathcal{T} \right] \right\}_{t \in [0, \mathcal{T}]} \), where \( \left\{ R^{\eta}_t \right\}_{t \in [0, \mathcal{T}]} \) is the logarithmic return process given by Eq.~(\ref{Rt}) corresponding to the strategy \( \eta \). For a fixed \( \eta \), the representative fund manager does not influence \( \overline{R} \), as they are just one fund manager in a continuum.

For any admissible strategy \( \pi \), define the excess logarithmic return process \( \{ Z_t := R^{\pi}_t - \theta \overline{R}_t \}_{t \in [0, \mathcal{T}]} \). The objective of the representative fund manager is to maximize the expected exponential utility of the excess logarithmic return at the terminal time \( Z_\mathcal{T} = R_\mathcal{T}^\pi - \theta \overline{R}_\mathcal{T} \):
\begin{eqnarray}
	J\left(\pi\right):=\mathbb{E}\left\{-\exp\left[-\frac{1}{\delta}Z_\mathcal{T}\right]\right\}.\label{Jpi}
\end{eqnarray} }

\begin{definition}\label{MFE}
		
	{
An admissible strategy \( \pi^* \) is a mean-field equilibrium (MFE) if it is optimal for the optimization problem of maximizing \( J(\pi) \), as shown in Eq.~(\ref{Jpi}), with \( \overline{R}_t = \mathbb{E}\left[ R^{\pi^*}_t \mid \mathcal{F}^{B}_\mathcal{T} \right] \) for \( t \in [0, \mathcal{T}] \).

A constant MFE is an \( \mathcal{F}^{MF}_0 \)-measurable two-dimensional random variable \( \pi^* \) such that \( \pi_t = \pi^* \) for \( t \in [0, \mathcal{T}] \), and this is also a MFE.
    }
\end{definition}

\begin{theorem}\label{the2}
	There exists a unique constant MFE, given by
	\begin{eqnarray}
		\alpha^*&=&\frac{\delta\widehat{\sigma}\left(\mu\widehat{\sigma}-\widehat{\mu}\sigma\right)}{\left(1+\delta\right)\sigma^2\widehat{\nu}^2}+\frac{\mu\delta+\theta\mu L}{\left(1+\delta\right)\sigma^2},\label{alpha1}\\
		\beta^{*}&=&\frac{\delta\left(\widehat{\mu}\sigma-\mu\widehat{\sigma}\right)}{\left(1+\delta\right)\sigma\widehat{\nu}^2}\label{beta1},\nonumber
	\end{eqnarray}
	where $L$ is defined by  Eq.~(\ref{L}).
\end{theorem}

\begin{proof}
	See Appendix \ref{proof2}.
\end{proof}
\begin{remark}
	The constant MFE $\left(\alpha^*, \beta^*\right)$ has the same expression as in Eqs.~(\ref{limit2}) and (\ref{limit1}), which represents the limit of the equilibrium strategy in Definition \ref{def1} as the number of fund managers approaches infinity.
\end{remark}
\begin{remark}
	By setting $\theta = 0$, we obtain the optimal investment ratio for an optimization problem without competition, which reveals some interesting phenomena. The { allocation fraction} of investment in the private risky asset, $\beta^*$, is unaffected by competition and reflects the optimal result of the optimization problem. Conversely, the { allocation fraction} invested in the public risk asset, $\alpha^*$, can be divided into two components: one that disregards competition and one that addresses competitive risk.
\end{remark}

\vskip 10pt
\section{{{\bf Mean-Variance criterion}}} \label{Mean-Variance criterion}
\vskip 5pt
In this section, we assume that each fund manager has a mean-variance preference for the excess logarithmic return, so the payoff of the $k$-th fund manager is given by { 
\begin{eqnarray}
	J_k\left(\pi^k;\pi^{-k}|t,r_1,\cdots,r_n\right):=\mathbb{E}_t\left(R^{k,\pi^k}_\mathcal{T}-\frac{\theta_k}{n}\sum_{i=1}^{n}R^{i,\pi^i}_\mathcal{T}\right)-\frac{\gamma_k}{2}\mathbb{V}_t\left(R^{k,\pi^k}_\mathcal{T}-\frac{\theta_k}{n}\sum_{i=1}^{n}R^{i,\pi^i}_\mathcal{T}\right),\nonumber 
\end{eqnarray} }
where $\gamma_k > 0$ is the risk aversion parameter for the $k$-th fund manager, and $\mathbb{E}_t$ and $\mathbb{V}_t$ represent the conditional expectation and conditional variance given the condition $\left(R^{1,\pi^1}_t, \cdots, R^{n,\pi^n}_t\right) = \left(r_1, \cdots, r_n\right)$.

The mean-variance criterion is time-inconsistent, meaning Bellman's principle does not hold and the standard stochastic dynamic programming method cannot be applied. To address this time-inconsistent game, we adopt the time-consistent Nash equilibrium strategy proposed by \cite{bjork2017time}. We require that, at any moment, a fund manager unilaterally changing the investment strategy will not achieve better results in the very short term. The time-consistent equilibrium and the constant time-consistent equilibrium are defined strictly below.

\begin{definition}\label{def2}
	{ A strategy $\left(\pi^{1},\cdots,\pi^{n}\right)$ is admissible if, for any $k=1,\cdots,n$, $\pi^k$ is $\{\mathcal{F}_t\}_{0\leq t\leq \mathcal{T}} $ progressively measurable and satisfies
	\begin{eqnarray}
		\mathbb{E}\int_{0}^{\mathcal{T}}\left(\alpha^{k2}_t+\beta^{k2}_t\right) \rd t < +\infty.
	\end{eqnarray}
	
	An admissible strategy $\left(\pi^{1,*},\cdots,\pi^{n,*}\right)$ is a time-consistent equilibrium if, for any $k=1,\cdots,n$, any fixed initial state $\left(t, r_1, \cdots, r_n\right) \in [0, \mathcal{T}] \times \mathbb{R}^n$, and any admissible strategy $u$ of the $k$-th fund manager, }
	\begin{eqnarray}
		\liminf_{h\rightarrow 0^+} \frac{J_k\left(\pi^{k,*};\pi^{-k,*}|t,r_1,\cdots,r_n\right)-J_k\left(\pi^{k,h,u};\pi^{-k,*}|t,r_1,\cdots,r_n\right)}{h}\ge 0,\nonumber
	\end{eqnarray}
	where $\pi^{k,h,u}$ is defined by { 
	\begin{eqnarray}
		\pi^{k,h,u}_s:=\left\{
		\begin{aligned}
			&u_s,\quad\qquad t\le s < t+h,\\
			&\pi^{k,*}_s, \quad t+h\le s\le \mathcal{T}
		\end{aligned}
		\right.\nonumber
	\end{eqnarray} }
	with any $h\in\mathbb{R}^{+}$.

A constant time-consistent equilibrium is a time-consistent equilibrium in which, for each $k$, $\alpha^{k,*}$ and $\beta^{k,*}$ are constants over time.
\end{definition}
\begin{theorem}\label{the3}
	There exists a unique constant time-consistent equilibrium, given by
	\begin{eqnarray}
		\alpha^{k,*}&=&\frac{\sigma_k\left(\mu\sigma_k-\mu_k\sigma\right)}{\left(1+\gamma_k-\frac{\gamma_k\theta_k}{n}\right)\sigma^2\nu_k^2}+\frac{\mu+\theta_k\sigma K}{\left(1+\gamma_k\right)\sigma^2},\quad k=1,\cdots,n,\label{alp2}\\
		\beta^{k,*}&=&\frac{\mu_k\sigma-\mu\sigma_k}{\left(1+\gamma_k-\frac{\gamma_k\theta_k}{n}\right)\sigma\nu_k^2},\quad k=1,\cdots,n,\label{bet2}
	\end{eqnarray}
	where
	\begin{eqnarray}
		K:= \frac{\frac{1}{n}\sum_{i=1}^{n}\frac{1}{1+\gamma_k}\frac{\mu}{\sigma}}{1-\frac{1}{n}\sum_{i=1}^{n}\frac{\gamma_k\theta_k}{1+\gamma_k}}.\label{K1}
	\end{eqnarray}
\end{theorem}

\begin{proof}
	See Appendix \ref{proof3}.
\end{proof}

\subsection{{{\bf Analysis about the constant time-consistent equilibrium $\left(\alpha^{k,*},\beta^{k,*}\right)$}}}
Similar to Section \ref{Exponential utility}, both the { allocation fraction} $\alpha^{k,*}$ of the investment in $S^*$ and the { allocation fraction} $\beta^{k,*}$ of the investment in $S^k$ are constant. The { allocation fraction} $\alpha^{k,*}$ is influenced by the parameters of other investment managers, including the competition weight parameter $\theta_i$ (for $i \neq k$) and the risk aversion coefficient $\delta_i$ (for $i \neq k$), but does not depend on the parameters associated with the risky assets $S^i$ (for $i \neq k$). In contrast, $\beta^{k,*}$ is not affected by the parameters of other investment managers but is influenced by its own competition weight parameter $\theta_k$. We also consider three cases: case 1, case 2, and case 3, as detailed in Section \ref{Exponential utility}.

{ Similar to Subsection \ref{infity_case}, we can formally obtain the result (the detailed process, including the definition of the type vector, type distribution, and weak convergence of measures, is omitted here):
\begin{eqnarray}
	\beta^{k,*}_{\infty}&:=&\lim\limits_{n\rightarrow +\infty}\beta^k =\frac{\mu_k\sigma-\mu\sigma_k}{\left(1+\gamma_k\right)\sigma\nu_k^2},\label{lim1}\\
	\alpha^{k,*}_{\infty}&:=&\lim\limits_{n\rightarrow +\infty}\alpha^{k} =\frac{\sigma_k\left(\mu\sigma_k-\mu_k\sigma\right)}{\left(1+\gamma_k\right)\sigma^2\nu_k^2}+\frac{\mu+\theta_k\mu R}{\left(1+\gamma_k\right)\sigma^2},\label{lim2}
\end{eqnarray}
where
\begin{eqnarray}
	R := \frac{\mathbb{E}\frac{1}{1+\gamma}}{1-\mathbb{E}\frac{\gamma\theta}{1+\gamma}},\label{R}
\end{eqnarray}
which is the same as Eqs.~ (\ref{limit1}) and (\ref{limit2}).}


{
The remainder of this section is organized as follows. First, we will examine the impact of one's own parameters. For more complex scenarios, numerical analysis and plotting will be required. Next, we will explore the influence of the parameters of other fund managers. Since the results in this section closely mirror those in Subsection \ref{analysis1}, we will only present the formulas and plots, providing explanations and notes only when necessary. Additionally, we will observe that the conclusions of Remark \ref{remark_new} also apply here.
}

The { allocation fraction} of the investment in $S^k$ denoted $\beta^{k,*}$, is also a modification of the equilibrium solution to the classical optimization problem without competition. We have
\begin{eqnarray}
	\frac{\partial\beta^{k,*}}{\partial\theta_k}&=&\frac{\gamma_k\left(\mu_k\sigma-\mu\sigma_k\right)}{n\left(1+\gamma_k-\frac{\gamma_k\theta_k}{n}\right)^2\sigma\nu_k^2},
\end{eqnarray}
which is the same as Eq.~ (\ref{partial1}).

In addition, 
\begin{eqnarray}
	&&\frac{\partial\beta^{k,*}}{\partial\mu}=\frac{-\sigma_k}{\left(1+\gamma_k-\frac{\gamma_k\theta_k}{n}\right)\sigma\nu_k^2},\qquad
	\frac{\partial\beta^{k,*}}{\partial\sigma}=\frac{\mu\sigma_k}{\left(1+\gamma_k-\frac{\gamma_k\theta_k}{n}\right)\sigma^2\nu_k^2},\nonumber\\
	&&\frac{\partial\beta^{k,*}}{\partial\mu_k}=\frac{1}{\left(1+\gamma_k-\frac{\gamma_k\theta_k}{n}\right)\nu_k^2},\qquad\;\;\;
	\frac{\partial\beta^{k,*}}{\partial\sigma_k}=\frac{-\mu}{\left(1+\gamma_k-\frac{\gamma_k\theta_k}{n}\right)\sigma\nu_k^2},\nonumber\\
	&&\frac{\partial\beta^{k,*}}{\partial\nu_k}=\frac{-2\left(\mu_k\sigma-\mu\sigma_k\right)}{\left(1+\gamma_k-\frac{\gamma_k\theta_k}{n}\right)\sigma\nu_k^3},\qquad
	\frac{\partial\beta^{k,*}}{\partial\gamma_k}=\frac{\left(1-\frac{\theta_k}{n}\right)\left(\mu_k\sigma-\mu\sigma_k\right)}{\left(1+\gamma_k-\frac{\gamma_k\theta_k}{n}\right)^2\sigma\nu_k^2},\nonumber
\end{eqnarray}
and
\begin{eqnarray}
	\frac{\partial\alpha^{k,*}}{\partial\mu}&=&\frac{\sigma_k^2}{\left(1+\gamma_k-\frac{\gamma_k\theta_k}{n}\right)\sigma^2\nu_k^2}+\frac{1}{\left(1+\gamma_k\right)\sigma^2}+\frac{\theta_k}{\left(1+\gamma_k\right)\sigma^2}\frac{\frac{1}{n}\sum_{j=1}^{n}\frac{1}{1+\gamma_j}}{1-\frac{1}{n}\sum_{j=1}^{n}\frac{\gamma_j\theta_j}{1+\gamma_j}},\nonumber\\
	\frac{\partial\alpha^{k,*}}{\partial\sigma}&=&\frac{-2\mu\sigma_k^2+\mu_k\sigma_k\sigma}{\left(1+\gamma_k-\frac{\gamma_k\theta_k}{n}\right)\sigma^3\nu_k^2}-\frac{2\mu}{\left(1+\gamma_k\right)\sigma^3}-\frac{2\mu\theta_k}{\left(1+\gamma_k\right)\sigma^3}\frac{\frac{1}{n}\sum_{j=1}^{n}\frac{1}{1+\gamma_j}}{1-\frac{1}{n}\sum_{j=1}^{n}\frac{\gamma_j\theta_j}{1+\gamma_j}},\nonumber\\
	\frac{\partial\alpha^{k,*}}{\partial\mu_k}&=&\frac{-\sigma_k}{\left(1+\gamma_k-\frac{\gamma_k\theta_k}{n}\right)\sigma\nu_k^2},\nonumber\\
	\frac{\partial\alpha^{k,*}}{\partial\sigma_k}&=&\frac{2\mu\sigma_k-\mu_k\sigma}{\left(1+\gamma_k-\frac{\gamma_k\theta_k}{n}\right)\sigma^2\nu_k^2},\nonumber\\
	\frac{\partial\alpha^{k,*}}{\partial\nu_k}&=&\frac{-2\sigma_k\left(\mu\sigma_k-\mu_k\sigma\right)}{\left(1+\gamma_k-\frac{\gamma_k\theta_k}{n}\right)\sigma^2\nu_k^3},\nonumber\\
	\frac{\partial\alpha^{k,*}}{\partial\theta_k}&=&\frac{\gamma_k\sigma_k\left(\mu\sigma_k-\mu_k\sigma\right)}{n\left(1+\gamma_k-\frac{\gamma_k\theta_k}{n}\right)^2\sigma^2\nu_k^2}+\frac{\mu}{\left(1+\gamma_k\right)\sigma^2}\frac{\frac{1}{n}\sum_{j=1}^{n}\frac{1}{1+\gamma_j}}{1-\frac{1}{n}\sum_{j=1}^{n}\frac{\gamma_j\theta_j}{1+\gamma_j}}\nonumber\\
	&&+\frac{\gamma_k\mu\theta_k}{n\left(1+\gamma_k\right)^2\sigma^2}\frac{\frac{1}{n}\sum_{j=1}^{n}\frac{1}{1+\gamma_j}}{\left(1-\frac{1}{n}\sum_{j=1}^{n}\frac{\gamma_j\theta_j}{1+\gamma_j}\right)^2}.
\end{eqnarray}
The signs of these partial derivatives are consistent with those in Eqs.~(\ref{partial2}) and (\ref{partial5}). Therefore, the sensitivity analysis of the parameters $\mu$, $\sigma$, $\mu_k$, $\sigma_k$, $\nu_k$, $\gamma_k$, and $\theta_k$ with respect to $\beta_k$, as well as the sensitivity analysis of the parameters $\mu$, $\sigma$, $\mu_k$, $\sigma_k$, $\nu_k$, and $\theta_k$ with respect to $\alpha_k$, are the same as those discussed in Section \ref{Exponential utility}.

\begin{figure}[htbp]
	\centering
	\includegraphics[scale=0.71]{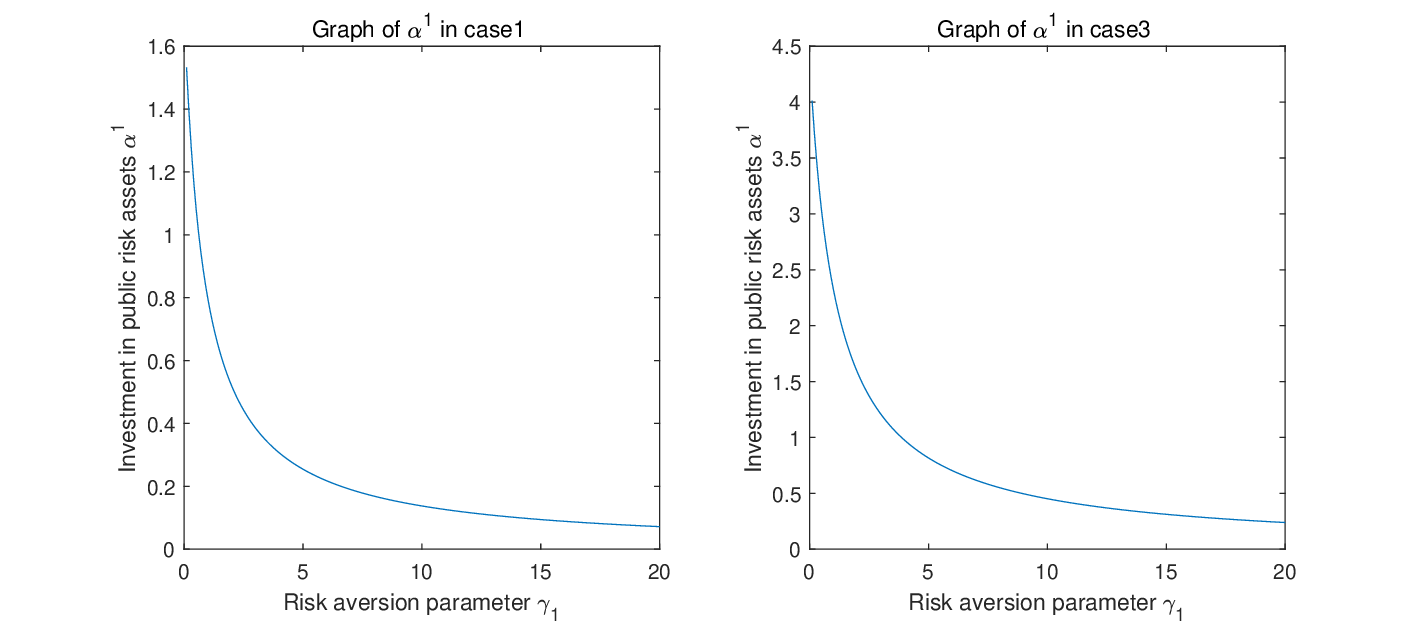}
	\caption{Investment in public risk assets $\alpha^{1,*}$ in case 1 and case 3}
	\label{Third4}
\end{figure}

\begin{figure}[htbp]
	\centering
	\includegraphics[scale=0.89]{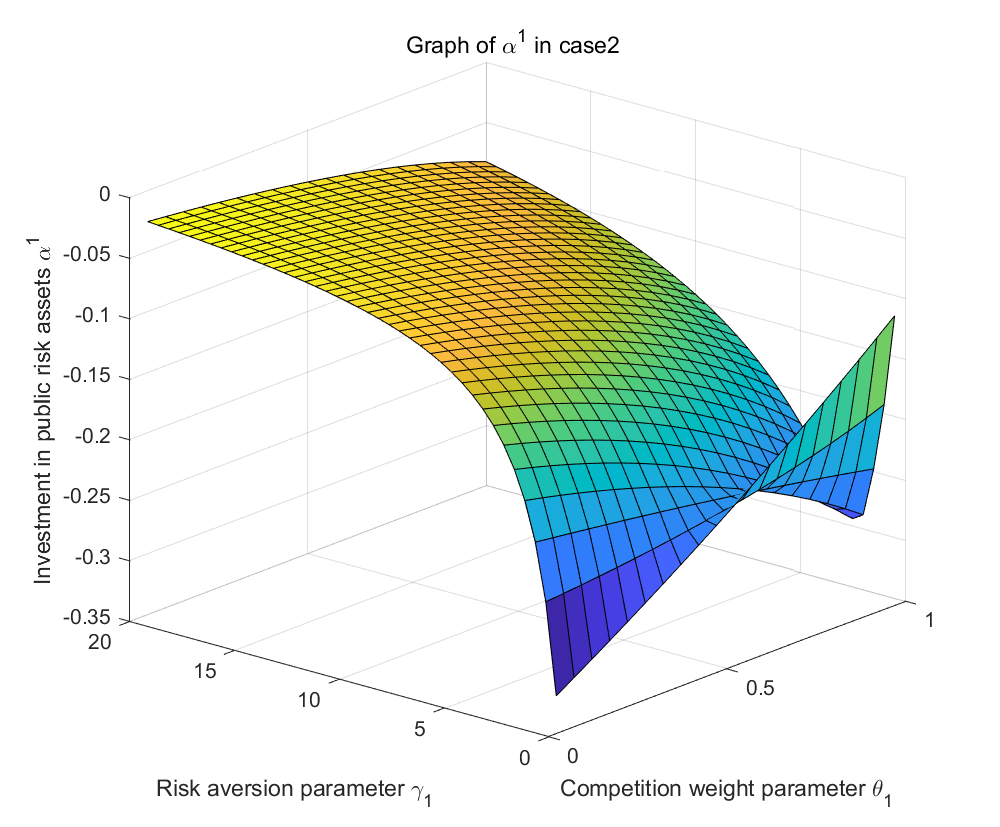}
	\caption{Investment in public risk assets $\alpha^{1,*}$ in case 2.}
	\label{Third5}
\end{figure} 

For the analysis of $\gamma_k$ in cases 1 and 3, and of both $\gamma_k$ and $\theta_k$ in case 2, we also consider two special cases: one with only two fund managers and another with an infinite number of fund managers. In the case of two fund managers, the analysis must be implemented numerically using graphical methods. We take $n=2,\left(\mu,\sigma\right)=\left(1,1\right),\theta_2=0.5,\gamma_2=5,\gamma_1\in\left(0,20\right)$, $\theta_1=0.5,\left(\mu_1,\sigma_1,\nu_1\right)=\left(2,3,3\right)$ in case 1, $\theta_1=0.5,\left(\mu_1,\sigma_1,\nu_1\right)=\left(2,-1,1\right)$ in case 3 and $\theta_1\in\left(0,1\right),\left(\mu_1,\sigma_1,\nu_1\right)=\left(5,2,2\right)$ in case 2. Fig.~\ref{Third4} shows that $\alpha^{1,*}$ increases as the risk aversion coefficient $\gamma_1$ decreases, regardless of whether it is case 1 or case 3. However, Fig.~\ref{Third5} illustrates the complex variations of $\alpha^{1,*}$ with respect to $\gamma_1$ and $\theta_1$ in case 2.

Unlike the analysis in Subsection \ref{analysis1}, the analysis for the case with an infinite number of managers differs from that with two managers, particularly in case 2. Note that, by taking the derivative and arranging it,
\begin{eqnarray}
	\frac{\partial\alpha^{k,*}_{\infty}}{\partial\theta_k}&=&\frac{\mu R}{\left(1+\gamma_k\right)\sigma^2}>0,\nonumber\\
	\frac{\partial\alpha^{k,*}_{\infty}}{\partial\gamma_k}&=& - \frac{1}{\left(1+\gamma_k\right)^2}\left[\frac{\sigma_k\left(\mu\sigma_k-\mu_k\sigma\right)}{\sigma^2\nu_k^2} + \frac{\mu\left(1+\theta_k R\right)}{\sigma^2}\right]\nonumber.
\end{eqnarray}
Regardless of whether it is case 1, case 2, or case 3, $\alpha^{k,*}_{\infty}$ increases as $\theta_k$ increases. In case 1 and case 3, $\alpha^{k,*}$ increases as $\gamma_k$ decreases. However, in case 2, $\alpha^{k,*}$ can either increase or decrease with $\gamma_k$, depending on $\theta_k$ and other parameters.

Furthermore, we consider the effect of other fund managers' parameters on the $k$-th fund manager's investment in $S^*$. For any $i \neq k$,
\begin{eqnarray}
	\frac{\partial\alpha^{k,*}}{\partial\theta_i}&=&\frac{\theta_k\mu}{\left(1+\gamma_k\right)\sigma^2}\frac{\frac{1}{n}\sum_{j=1}^{n}\frac{1}{1+\gamma_j}}{\left(1-\frac{1}{n}\sum_{j=1}^{n}\frac{\gamma_j\theta_j}{1+\gamma_j}\right)^2}\frac{\gamma_i}{n\left(1+\gamma_i\right)}>0\nonumber\\
	\frac{\partial\alpha^{k,*}}{\partial\delta_i}&=&-\frac{\theta_k\mu}{\left(1+\gamma_k\right)\sigma^2}\frac{1-\frac{1}{n}\sum_{j=1}^{n}\frac{\gamma_j\theta_j}{1+\gamma_j}-\frac{1}{n}\sum_{j=1}^{n}\frac{1}{1+\gamma_j}\gamma_i}{n\left(1+\gamma_i\right)^2\left(1-\frac{1}{n}\sum_{j=1}^{n}\frac{\gamma_j\theta_j}{1+\gamma_j}\right)^2}\nonumber\\
	&\le& -\frac{\theta_k\mu}{\left(1+\gamma_k\right)\sigma^2}\frac{1-\max_{j}\theta_j}{n\left(1+\gamma_i\right)^2\left(1-\frac{1}{n}\sum_{j=1}^{n}\frac{\gamma_j\theta_j}{1+\gamma_j}\right)^2} <0\nonumber.
\end{eqnarray} 
An increase in the competition weight parameter $\theta_i$ or a decrease in the risk aversion parameter $\gamma_i$ of any other fund manager $i$ will lead to an increase in the { allocation fraction} of the $k$-th fund manager's investment in the public risk asset $S^*$.
\subsection{{{\bf The MFG under mean-variance criterion}}}
In this subsection, we revise the financial model to construct the mean-field game (MFG) under the mean-variance criterion, assuming an infinite number of fund managers where the parameters of each fund manager are random variables. Similar to Subsection \ref{Themfg}, we assume that the probability space $\left(\Omega,\mathcal{F},\mathbb{P}\right)$ supports also another independent one-dimensional Brownian motion, $W$, as well as a type vector
\begin{eqnarray}
	\Phi = \left(\gamma,\theta,\widehat{\mu},\widehat{\sigma},\widehat{\nu}\right)\nonumber
\end{eqnarray}
with the values in the space $\left(0,+\infty\right)\times\left[0,1\right)\times\left(0,+\infty\right)\times\mathbb{R}\times\left(0,+\infty\right)$, which is independent of $W$ and $B$. { 

Let $\mathcal{F}^{MF} = \left\{\mathcal{F}^{MF}_t\right\}_{t \in [0, \mathcal{T}]}$ denote the smallest filtration satisfying the usual assumptions under which $\Phi$ is $\mathcal{F}^{MF}_0$-measurable, and both $W$ and $B$ are adapted. Additionally, let $\mathcal{F}^{B} = \left\{\mathcal{F}^{B}_t\right\}_{t \in [0, \mathcal{T}]}$ denote the natural filtration generated by $B$.

The logarithmic return $\left\{R_t^\pi\right\}_{t \in [0, \mathcal{T}]}$ of any fund manager with type vector $\Phi$ evolves as given by Eq.~(\ref{Rt}), and a strategy $\pi = \left(\alpha, \beta\right)$ is admissible if it is $\mathcal{F}^{MF}$ progressively measurable and satisfies Eq.~(\ref{admissible}). For an admissible strategy $\eta$, define an $\mathcal{F}^{B}_\mathcal{T}$-measurable random process $\{\overline{R}_t = \mathbb{E}\left[R^{\eta}_t|\mathcal{F}^{B}_\mathcal{T}\right]\}_{t \in [0, \mathcal{T}]}$, where $\{R^{\eta}_t\}_{t \in [0, \mathcal{T}]}$ is the logarithmic return process in Eq.~(\ref{Rt}) corresponding to the strategy $\eta$. 

For a fixed $\eta$, the payoff of the manager with type vector $\Phi$ and admissible strategy $\pi$ is given by:
\[
J\left(\pi|(t,z)\right) := \mathbb{E}_t\left(Z_\mathcal{T}\right) - \frac{\gamma}{2} \mathbb{V}_t\left(Z_\mathcal{T}\right),
\]
where the process $\{Z_t\}_{t \in [0, \mathcal{T}]}$ is defined as $Z_t := R^{\pi}_t - \theta \overline{R}_t$, and $\mathbb{E}_t$ and $\mathbb{V}_t$ represent the conditional expectation and variance given the condition $Z_t = z$.
}

 Now we obtain a time-consistent MFE for this MFG.
 \begin{definition}\label{defination2}
 	{ An admissible strategy $\pi^*$ is a time-consistent MFE for this MFG if $\overline{R}_t = \mathbb{E}\left[R^{\pi^*}_t|\mathcal{F}^{B}_\mathcal{T}\right], t\in\left[0,\mathcal{T}\right]$, and for any fixed initial state $\left(t,z\right)\in\left[0,\mathcal{T}\right]\times\mathbb{R}$, and any admissible strategy $u$, }
 	\begin{eqnarray}
 		\liminf_{h\rightarrow 0^+} \frac{J\left(\pi^{*}|\left(t,z\right)\right)-J\left(\pi^{h,u}|\left(t,z\right)\right)}{h}\ge 0,\nonumber
 	\end{eqnarray}
 	where $\pi^{h,u}$ is defined by { 
 	\begin{eqnarray}
 		\pi^{h,u}_s:=\left\{
 		\begin{aligned}
 			&u_s,\quad\qquad t\le s < t+h,\\
 			&\pi^{*}_s, \quad\quad t+h\le s\le \mathcal{T}
 		\end{aligned}
 		\right.\nonumber
 	\end{eqnarray} }
 	with any $h\in\mathbb{R}^{+}$.
 		
 	A constant time-consistent MFE is a time-consistent equilibrium in which $\alpha^*$ and $\beta^*$ are both constant random variables over time.
 \end{definition}

\begin{theorem}\label{the4}
	There exists a unique constant time-consistent MFE, given by
	\begin{eqnarray}
		\alpha^* &=&\frac{\widehat{\sigma}\left(\mu\widehat{\sigma}-\widehat{\mu}\sigma\right)}{\left(1+\gamma\right)\sigma^2\widehat{\nu}^2}+\frac{\mu+\theta\mu R}{\left(1+\gamma\right)\sigma^2},\label{alpha2}\\
		\beta^* &=& \frac{\widehat{\mu}\sigma-\mu\widehat{\sigma}}{\left(1+\gamma\right)\sigma\widehat{\nu}^2},\label{beta2}
	\end{eqnarray}
	where $R$ is shown in Eq.~(\ref{R}).
\end{theorem}

\begin{proof}
	See Appendix \ref{proof4}.
\end{proof}
\begin{remark}
	The constant time-consistent MFE $\left(\alpha^*, \beta^*\right)$ has the same expression as given in Eqs.~(\ref{lim2}) and (\ref{lim1}), representing the limit of the time-consistent equilibrium strategy in Definition \ref{def2} as the number of fund managers approaches infinity.
\end{remark}
\begin{remark}
Similar to the MFE analysis in the previous section, the { allocation fraction} of investment in the private risky asset $\beta^*
$ remains unaffected by competition and represents the optimal solution to the optimization problem. The  { allocation fraction} in the public risk asset $\alpha^*$ can be divided into two components: one that ignores competition and another that addresses competitive risk.
\end{remark}

\vskip 10pt
\section{\bf Conclusion}
This paper uses excess logarithmic returns as the evaluation criterion for funds and constructs both the $n$-player game and the MFG of portfolio management under the relative performance for fund managers using two criteria. In the first criterion, the fund manager's risk aversion level is assumed to be independent of excess logarithmic returns, with the exponential utility function selected in relation to the excess logarithmic return. In the second criterion, the fund manager is assumed to follow a mean-variance criterion, for which we define a time-consistent equilibrium. Each fund manager not only has a unique risky asset but can also invest in a public risk-free asset and a public risky asset. We rigorously define the equilibrium for each game and provide the explicit solution for the portfolio.

Whether considering fund managers with exponential utility or those with a mean-variance criterion, and whether analyzing the $n$-player game or the MFG, the introduction of excess logarithmic returns results in a constant investment ratio. The MFE of the MFG aligns with the equilibrium limit in the $n$-player game as $n$ approaches infinity. The inclusion of public risk assets significantly influences the outcomes, with the impact of competition mainly reflected in the investment in public risk assets. In the $n$-player game, the  { allocation fraction} in private risky assets is only affected by the competition weight parameter and the number of competitors $n$, and is not influenced by the parameters of other fund managers. The { allocation fraction} in the public risky asset is affected by the competition weight parameter and the risk aversion coefficient of other fund managers but not by the parameters of other fund managers' risk assets. In the MFG, the { allocation fraction} in private risk assets mirrors the optimal investment problem without competition and is not influenced by competition-related parameters. The  { allocation fraction} in public risky assets can be divided into two components: one that ignores competition and one that hedges competitive risk. The sensitivity analysis of the equilibrium is demonstrated through derivative analysis or numerical graphics.

\vskip 15pt
\paragraph{\bf Acknowledgements.} The authors thank the members of the group of Actuarial Sciences and Mathematical Finance at the Department of Mathematical Sciences, Tsinghua University for their feedback and useful conversations. 
\vskip 10pt
\paragraph{\bf\large Compliance with Ethical Standards}
\vskip 15pt
\paragraph{\bf Disclosure of potential conflicts of interest} All authors declare that they have no conflict of interest.
\vskip 15pt
\paragraph{\bf Funding} This study was funded by the National Natural Science Foundation of China (Grant Nos. 12371477, 12271290), the MOE Project of Key Research Institute of Humanities and Social Sciences (22JJD910003).
\vskip 15pt
\paragraph{\bf Ethical approval} This article does not contain any studies with human participants or animals performed by any of the authors

\appendix
\renewcommand{\theequation}{\thesection.\arabic{equation}}
\vskip 15pt

\section{Proof of Theorem \ref{the1}.}\label{proof1}

\begin{proof}
	First, the strategies of the other fund managers are fixed and only the $k$-th fund manager's optimization problem is considered. The $k$-th fund manager's strategy $\pi^{k,*}$ in equilibrium, is optimal for this optimization problem. Notice that { 
	\begin{eqnarray}
		R^k_\mathcal{T}-\frac{\theta_k}{n}\sum_{i=1}^{n}R^i_\mathcal{T} = \left(1-\frac{\theta_k}{n}\right)R^k_\mathcal{T}-\theta_k\frac{1}{n}\sum_{i\neq k}R^i_\mathcal{T},\label{eq1}
	\end{eqnarray} }
	and define
	\begin{eqnarray}
		Y_t:=\frac{1}{n}\sum_{i\neq k}R^i_t.\label{eq2}
	\end{eqnarray}
Then, the optimization problem of the $k$-th fund manager becomes {
	\begin{eqnarray}
		v\left(t,r,y\right):=\sup_{\pi^k}\mathbb{E}\left\{-\exp\left[-\frac{1}{\delta_k}\left(\left(1-\frac{\theta_k}{n}\right)R^k_\mathcal{T}-\theta_kY_\mathcal{T}\right)\right]\bigg|\left(R^k_t,Y_t\right)=\left(r,y\right)\right\}.\nonumber
	\end{eqnarray}
The processes $\{R^k_t\}_{t\in\left[0,\mathcal{T}\right]}$ and $\{Y_t\}_{t\in\left[0,\mathcal{T}\right]}$ evolve as }
\begin{eqnarray}
\!\! \!\!	\!\!	\!\!			\rd R^k_t &=& \left(\kappa + \pi^{kT}_t \tilde{\mu}_k   - \frac{1}{2}\pi^{kT}_t\tilde{\sigma}_k\tilde{\sigma}_k^T\pi^{k}_t - \frac{1}{2}\pi^{kT}_t\tilde{\nu}_k\tilde{\nu}_k^T\pi^{k}_t\right) \rd t +  \pi^{kT}_t\tilde{\sigma}_k \rd B_t + \pi^{kT}_t\tilde{\nu}_k \rd W^k_t,\label{r1}\\
		\rd Y_t &=& \left[\frac{n-1}{n}\kappa + \frac{1}{n}\sum_{i\neq k}\pi^{iT}_t \tilde{\mu}_i   - \frac{1}{2n}\sum_{i\neq k}\pi^{iT}_t\left(\tilde{\sigma}_i\tilde{\sigma}_i^T+\tilde{\nu}_i\tilde{\nu}_i^T\right)\pi^{i}_t\right] \rd t \nonumber \\
		&&+  \left(\frac{1}{n}\sum_{i\neq k}\pi^{iT}_t\tilde{\sigma}_i\right) \rd B_t + \frac{1}{n}\sum_{i\neq k}\pi^{iT}_t\tilde{\nu}_i \rd W^i_t, \label{y1}
	\end{eqnarray}
	where 
	$$ \tilde{\mu}_k := \begin{pmatrix}
		\mu \\ \mu_k
	\end{pmatrix},\quad \tilde{\sigma}_k := \begin{pmatrix}
		\sigma \\ \sigma_k
	\end{pmatrix}, \quad \tilde{\nu}_k := \begin{pmatrix}
		0 \\ \nu_k
	\end{pmatrix} .$$
	
	Assume that the strategies of other fund managers are constant, i.e., $\pi^i_t\equiv \pi^i, \forall i \neq k$.  Define
	\begin{eqnarray}
		A &=& \frac{n-1}{n}\kappa + \frac{1}{n}\sum_{i\neq k}\pi^{iT} \tilde{\mu}_i   - \frac{1}{2n}\sum_{i\neq k}\pi^{iT}\left(\tilde{\sigma}_i\tilde{\sigma}_i^T+\tilde{\nu}_i\tilde{\nu}_i^T\right)\pi^{i},\nonumber\\
		B &=& \frac{1}{n}\sum_{i\neq k}\pi^{iT}\tilde{\sigma}_i,\label{B1}\\
		C &=& \frac{1}{n^2}\sum_{i\neq k}\pi^{iT}\tilde{\nu}_i\tilde{\nu}_{iT}\pi^i.\nonumber
	\end{eqnarray}
Then	the HJB equation of this optimization problem is {
	\begin{eqnarray}
		&&v_t+\max_{\pi\in\mathbb{R}^2}\left\{\left[\kappa+\pi^T\tilde{\mu}_k-\frac{1}{2}\pi^T\left(\tilde{\sigma}_k\tilde{\sigma}_k^T+\tilde{\nu}_k\tilde{\nu}_k^T\right)\pi\right]v_r+\frac{1}{2}\pi^T\left(\tilde{\sigma}_k\tilde{\sigma}_k^T+\tilde{\nu}_k\tilde{\nu}_k^T\right)\pi v_{rr}\right.\nonumber\\
		&&\qquad\qquad\qquad\qquad\qquad\qquad\qquad+\left.Av_y+B\pi^T\tilde{\sigma}_kv_{ry}+\frac{1}{2}\left(B^2+C\right)v_{yy}\right\}=0,\nonumber\\
		&&v(\mathcal{T},r,y) = -\exp\left[-\frac{1}{\delta_k}\left(\left(1-\frac{\theta_k}{n}\right)r-\theta_ky\right)\right].\nonumber
	\end{eqnarray} }
	For this classical optimization problem, the maximum point
	\begin{eqnarray}\label{pi1}
		\pi^{k,*}=-\frac{1}{v_{rr}-v_r}\left(\tilde{\sigma}_k\tilde{\sigma}_k^T+\tilde{\nu}_k\tilde{\nu}_k^T\right)^{-1}\left(\tilde{\mu}_kv_r+B\tilde{\sigma}_kv_{ry}\right)
	\end{eqnarray}
	is the optimal strategy of the $k$-th fund manager, then the HJB equation becomes { 
	\begin{eqnarray} \label{HJB1}
		&&v_t \! + \! Av_y \! + \! \frac{1}{2}\left(B^2 \! + \! C\right)v_{yy} \! - \! \frac{1}{2\left(v_{rr} \! - \! v_r\right)}\left(\tilde{\mu}_kv_r+B\tilde{\sigma}_kv_{ry}\right)^T\!\left(\tilde{\sigma}_k\tilde{\sigma}_k^T+\tilde{\nu}_k\tilde{\nu}_k^T\right)^{-1}\!\left(\tilde{\mu}_kv_r+B\tilde{\sigma}_kv_{ry}\right)=0,\nonumber\\
		&&v(\mathcal{T},r,y) = -\exp\left[-\frac{1}{\delta_k}\left(\left(1-\frac{\theta_k}{n}\right)r-\theta_ky\right)\right].
	\end{eqnarray} }
	The HJB equation (\ref{HJB1}) has a unique solution of the form $v=f\left(t\right)g\left(r,y\right)$, which means that $\pi^{k,*}$ is a constant, and the solution is { 
	\begin{eqnarray}
		v=-\exp\left[-\frac{1}{\delta_k}\left(\left(1-\frac{\theta_k}{n}\right)r-\theta_ky\right)-\rho\left(\mathcal{T}-t\right)\right]\nonumber
	\end{eqnarray} }
	with
	\begin{eqnarray}
		\rho=-A\frac{\theta_k}{\delta_k}-\frac{1}{2}\left(B^2+C\right)\frac{\theta_k^2}{\delta_k^2}+\frac{1-\frac{\theta_k}{n}}{2\left(1-\frac{\theta_k}{n}+\delta_k\right)}\left(\tilde{\mu}_k+B\frac{\theta_k}{\delta_k}\tilde{\sigma}_k\right)^T\!\left(\tilde{\sigma}_k\tilde{\sigma}_k^T+\tilde{\nu}_k\tilde{\nu}_k^T\right)^{-1}\!\left(\tilde{\mu}_k+B\frac{\theta_k}{\delta_k}\tilde{\sigma}_k\right).\nonumber
	\end{eqnarray}

	Define two matrices
	\begin{eqnarray}
  M &:=& \tilde{\sigma}_k\tilde{\sigma}_k^T+\tilde{\nu}_k\tilde{\nu}_k^T=\begin{pmatrix}
			\sigma^2 & \sigma\sigma_k\\\sigma\sigma_k & \sigma_k^2+\nu_k^2
		\end{pmatrix},\label{M1}\\
		N &:=&\begin{pmatrix}
			\sigma_k^2+\left(1-\frac{\theta_k}{n\left(1+\delta_k\right)}\right)\nu_k^2 & -\sigma\sigma_k\\-\sigma\sigma_k & \sigma^2
		\end{pmatrix},\nonumber
	\end{eqnarray}
	and
	\begin{eqnarray}
		E:=\frac{1}{n}\sum_{i\neq k}\pi^{iT}\tilde{\sigma}_i+\frac{1}{n}\pi^{*T}\tilde{\sigma}_k.\label{E1}
	\end{eqnarray}
	 Then the maximum point (\ref{pi1}) can be rewritten as
	\begin{eqnarray}
		\pi^{k,*}&=&\frac{1}{1-\frac{\theta_k}{n}+\delta_k}M^{-1}\left(\delta_k\tilde{\mu}_k+B\theta_k\tilde{\sigma}_k\right)\nonumber\\
		&=&\frac{1}{1-\frac{\theta_k}{n}+\delta_k}M^{-1}\left(\delta_k\tilde{\mu}_k+E\theta_k\tilde{\sigma}_k\right)-\frac{\theta_k}{n\left(1-\frac{\theta_k}{n}+\delta_k\right)}M^{-1}\tilde{\sigma}_k\tilde{\sigma}_k^T\pi^{k,*}.\nonumber
	\end{eqnarray}
	Thus
	\begin{eqnarray}
		\pi^{k,*}&=&\left[\left(1-\frac{\theta_k}{n}+\delta_k\right)I+\frac{\theta_k}{n}M^{-1}\tilde{\sigma}_k\tilde{\sigma}_k^T\right]^{-1}M^{-1}\left(\delta_k\tilde{\mu}_k+E\theta_k\tilde{\sigma}_k\right)\nonumber\\
		&=& \left[M\begin{pmatrix}
			1+\delta_k & \frac{\theta_k\sigma_k}{n\sigma}\\0 & 1+\delta_k-\frac{\theta_k}{n}
		\end{pmatrix}\right]^{-1}\left(\delta_k\tilde{\mu}_k+E\theta_k\tilde{\sigma}_k\right)\nonumber\\
		&=& \frac{1}{\left(1+\delta_k-\frac{\theta_k}{n}\right)\sigma^2\nu_k^2}N\left(\delta_k\tilde{\mu}_k+E\theta_k\tilde{\sigma}_k\right).\label{0pi2}
	\end{eqnarray}

{
	A standard verification argument as in Theorem 3.5.2 of \cite{pham2009continuous} or Theorem III.8.1 of \cite{fleming2006controlled} establishes the optimality of the strategies  $\pi^{k,*}$ derived above. Uniqueness follows from the same verification arguments using the strict concavity of the objective functions. 
}
	
	Next, consider that each fund manager achieves the optimum, that is, $\left(\pi_1^*,\cdots,\pi_n^*\right)$ is an equilibrium. { From the previous discussion, we can conclude that $\pi^{k,*}$ takes the form given in Eq.~(\ref{0pi2}) for any $k$.} Then $E=\frac{1}{n}\sum_{i=1}^{n}\tilde{\sigma}_i^T\pi^{i,*}$, and we have
	\begin{eqnarray}
		\tilde{\sigma}_k^T\pi^{k,*}=\frac{\delta_k}{1+\delta_k}\frac{\mu}{\sigma}+\frac{\theta_k}{1+\delta_k}E,\quad k=1,\cdots,n.\nonumber
	\end{eqnarray}
	Further,
	\begin{eqnarray}
		E=\frac{1}{n}\sum_{i=1}^{n}\tilde{\sigma}_i^T\pi^{i,*}=\frac{1}{n}\sum_{i=1}^{n}\frac{\delta_k}{1+\delta_k}\frac{\mu}{\sigma}+\frac{1}{n}\sum_{i=1}^{n}\frac{\theta_k}{1+\delta_k}E.\nonumber
	\end{eqnarray}
	Note that $\frac{1}{n}\sum_{i=1}^{n}\frac{\theta_k}{1+\delta_k} <1 $ and then $E$ has the same representation as Eq.~(\ref{D1}), i.e., $E=D$.
	
	Finally, the equilibrium is given as Eqs.~(\ref{alp1}) and (\ref{bet1}), and the uniqueness follows from the above proof.
	
\end{proof}
\section{Proof of Theorem \ref{the2}.}\label{proof2}

\begin{proof}
	First, for some $\mathcal{F}^{MF}_0$ measurable two dimensional random variable $\eta$ with $\mathbb{E}\Vert\eta\Vert^2 < +\infty$, { define $\overline{R}_t := \mathbb{E}\left[R^{\eta}_t|\mathcal{F}^{B}_\mathcal{T}\right],t\in\left[0,\mathcal{T}\right]$. } Because $\phi$, $B$, and $W$ are independent, we have
	\begin{eqnarray}
		\rd \overline{R}_t &=& \mathbb{E}\left[\kappa + \eta^{T} \tilde{\mu}   - \frac{1}{2}\eta^{T}\tilde{\sigma}\tilde{\sigma}^T\eta - \frac{1}{2}\eta^{T}\tilde{\nu}\tilde{\nu}^T\eta\right] \rd t + \mathbb{E}\left[\eta^{T}\tilde{\sigma}\right] \rd B_t,\nonumber\\
		\overline{R}_0 &=& 0,\nonumber
	\end{eqnarray}
	where 
	$$ \tilde{\mu} := \begin{pmatrix}
		\mu \\ \widehat{\mu}
	\end{pmatrix},\quad \tilde{\sigma} := \begin{pmatrix}
		\sigma \\ \widehat{\sigma}
	\end{pmatrix}, \quad \tilde{\nu} := \begin{pmatrix}
		0 \\ \widehat{\nu}
	\end{pmatrix}. $$
	
	{ For some admissible strategy $\pi$, the excess logarithmic return process $\{Z_t\}_{t\in\left[0,\mathcal{T}\right]}$ evolves as }
	\begin{eqnarray}
		\rd Z_t &=& \left(\kappa + \pi^{T}_t \tilde{\mu}  - \frac{1}{2}\pi^{T}_t\tilde{\sigma}\tilde{\sigma}^T\pi_t - \frac{1}{2}\pi^{T}_t\tilde{\nu}\tilde{\nu}^T\pi_t-\theta\mathbb{E}\left[\kappa + \eta^{T} \tilde{\mu}   - \frac{1}{2}\eta^{T}\tilde{\sigma}\tilde{\sigma}^T\eta - \frac{1}{2}\eta^{T}\tilde{\nu}\tilde{\nu}^T\eta\right] \right) \rd t \nonumber\\
		&&+  \left(\pi^{T}_t\tilde{\sigma}-\theta\mathbb{E}\left[\eta^{T}\tilde{\sigma}\right]\right) \rd B_t + \pi^{T}_t\tilde{\nu} \rd W_t,\nonumber\\
		Z_0 &=& 0,\label{Zt}
	\end{eqnarray}

	{ Assume that $\pi^*$  is optimal for the optimization problem of maximizing $\mathbb{E}\left\{-\exp\left[-\frac{1}{\delta}\left(Z_\mathcal{T}\right)\right]\right\}$, the value function 
	\begin{eqnarray}
		v\left(t,z\right):=\sup_{\pi}\mathbb{E}\left\{-\exp\left[-\frac{1}{\delta}\left(Z_\mathcal{T}\right)\right]\bigg|Z_t=z\right\}\nonumber
	\end{eqnarray} }
	is the solution of the HJB equation { 
	\begin{eqnarray}
		&&v_t+\max_{\pi\in\mathbb{R}^2}\left\{\left[\kappa+\pi^T\tilde{\mu}-\frac{1}{2}\pi^T\left(\tilde{\sigma}\tilde{\sigma}^T+\tilde{\nu}\tilde{\nu}^T\right)\pi-\theta\mathbb{E}\left[\kappa + \eta^{T} \tilde{\mu}   - \frac{1}{2}\eta^{T}\tilde{\sigma}\tilde{\sigma}^T\eta - \frac{1}{2}\eta^{T}\tilde{\nu}\tilde{\nu}^T\eta\right]\right]v_z\right.\nonumber\\
		&&\qquad\qquad\qquad+\left.\frac{1}{2}\left[\pi^T\tilde{\nu}\tilde{\nu}^T\pi+\left(\tilde{\sigma}^T\pi-\theta\mathbb{E}\left[\tilde{\sigma}^T\eta\right]\right)^2\right] v_{zz}\right\}=0,\nonumber\\
		&&v(\mathcal{T},z) = -\exp\left(-\frac{1}{\delta}z\right).\nonumber
	\end{eqnarray} }
The maximum point
	\begin{eqnarray}\label{pi2}
		\pi^*=-\frac{1}{v_{zz}-v_z}\left(\tilde{\sigma}\tilde{\sigma}^T+\tilde{\nu}\tilde{\nu}^T\right)^{-1}\left(\tilde{\mu}v_z-\theta\mathbb{E}\left[\tilde{\sigma}^T\eta\right]\tilde{\sigma}v_{zz}\right)\nonumber
	\end{eqnarray}
	is the optimal strategy, and the HJB equation becomes { 
	\begin{eqnarray} 
		&&v_t  +  \left(\kappa-\theta\mathbb{E}\left[\kappa + \eta^{T} \tilde{\mu}   - \frac{1}{2}\eta^{T}\tilde{\sigma}\tilde{\sigma}^T\eta - \frac{1}{2}\eta^{T}\tilde{\nu}\tilde{\nu}^T\eta\right]\right)v_z  +   \frac{1}{2}\theta^2\mathbb{E}^2\left[\tilde{\sigma}^T\eta\right]v_{zz}  \nonumber\\
		&& \quad-  \frac{1}{2\left(v_{zz}  -  v_z\right)}\left(\tilde{\mu}v_z-\theta\mathbb{E}\left[\tilde{\sigma}^T\eta\right]\tilde{\sigma}v_{zz}\right)^T\left(\tilde{\sigma}\tilde{\sigma}^T+\tilde{\nu}\tilde{\nu}^T\right)^{-1}\left(\tilde{\mu}v_z-\theta\mathbb{E}\left[\tilde{\sigma}^T\eta\right]\tilde{\sigma}v_{zz}\right)=0,\nonumber\\
		&&v(\mathcal{T},z) = -\exp\left(-\frac{1}{\delta}z\right).\label{HJB2}
	\end{eqnarray} }
	The HJB equation (\ref{HJB2}) has a unique solution of the form $v=f\left(t\right)g\left(r,y\right)$, which means that $\pi^*$ is a constant, and the solution is { 
	\begin{eqnarray}
		v=-\exp\left[-\frac{z}{\delta}-\rho\left(\mathcal{T}-t\right)\right]\nonumber
	\end{eqnarray} }
	with
	\begin{eqnarray}
		\rho&=&-\frac{\kappa+\theta\mathbb{E}\left[\kappa + \eta^{T} \tilde{\mu}   - \frac{1}{2}\eta^{T}\tilde{\sigma}\tilde{\sigma}^T\eta - \frac{1}{2}\eta^{T}\tilde{\nu}\tilde{\nu}^T\eta\right]}{\delta}-\frac{\theta^2\mathbb{E}^2\left[\tilde{\sigma}^T\eta\right]}{2\delta^2}\nonumber\\
		&&+\frac{1}{2\left(1+\delta\right)}\left(\tilde{\mu}+\frac{\theta}{\delta}\mathbb{E}\left[\tilde{\sigma}^T\eta\right]\tilde{\sigma}\right)^T\left(\tilde{\sigma}\tilde{\sigma}^T+\tilde{\nu}\tilde{\nu}^T\right)^{-1}\left(\tilde{\mu}+\frac{\theta}{\delta}\mathbb{E}\left[\tilde{\sigma}^T\eta\right]\tilde{\sigma}\right).\nonumber
	\end{eqnarray}
	Therefore, the maximum point (\ref{pi2}) can be organized as
	\begin{eqnarray}
		\pi^*&=&\frac{1}{1+\delta}\left(\tilde{\sigma}\tilde{\sigma}^T+\tilde{\nu}\tilde{\nu}^T\right)^{-1}\left(\delta\tilde{\mu}+\theta\mathbb{E}\left[\tilde{\sigma}^T\eta\right]\tilde{\sigma}\right).\nonumber
	\end{eqnarray}

{
	Similar to Appendix \ref{proof1} [Proof of Theorem \ref{the1}], if follows from a standard verification theorem that there can be at most one classical solution of this HJB equation (\ref{HJB2}); see, e.g., the proof of Theorem 3.5.2 in \cite{pham2009continuous} (or Theorem III.8.1 in \cite{fleming2006controlled}) for the standard verification argument.
}

	Next, we take $\eta=\pi^*$, so that $\pi^*=\begin{pmatrix}
		\alpha^* \\ \beta^*
	\end{pmatrix}$ is a MFE defined in Definition \ref{MFE}, as such, 
	\begin{eqnarray}
		\tilde{\sigma}^T\pi^* = \frac{1}{1+\delta}\tilde{\sigma}^T\left(\tilde{\sigma}\tilde{\sigma}^T+\tilde{\nu}\tilde{\nu}^T\right)^{-1}\left(\delta\tilde{\mu}+\theta\mathbb{E}\left[\tilde{\sigma}^T\pi^*\right]\tilde{\sigma}\right)=\frac{\delta}{1+\delta}\frac{\mu}{\sigma}+\frac{\theta}{1+\delta}\mathbb{E}\left[\tilde{\sigma}^T\pi^*\right],\nonumber
	\end{eqnarray}
	then
	\begin{eqnarray}
		\mathbb{E}\left[\tilde{\sigma}^T\pi^*\right] =\mathbb{E}\left[\frac{\delta}{1+\delta}\right]\frac{\mu}{\sigma}+\mathbb{E}\left[\frac{\theta}{1+\delta}\right]\mathbb{E}\left[\tilde{\sigma}^T\pi^*\right].\nonumber
	\end{eqnarray}
	Thus $\mathbb{E}\left[\tilde{\sigma}^T\pi^*\right] = L\frac{\mu}{\sigma}$,
	where $L$ is shown in Eq.~(\ref{L}). 
	After further analysis, we  see that $\alpha^*$ and $\beta^*$ have the forms Eqs.~(\ref{alpha1}) and (\ref{beta1}).  The uniqueness of these forms follows from the preceding proof.
\end{proof}

\section{Proof of Theorem \ref{the3}.}\label{proof3}

\begin{proof}
	
	First, the strategies of the other fund managers are assumed to be fixed, and the focus is on optimizing the $k$-th fund manager's problem. Using Eqs.~(\ref{eq1}), (\ref{eq2}), (\ref{r1}), and (\ref{y1}), we have:
	\begin{eqnarray}
		\rd\left[R^k_t-\frac{\theta_k}{n}\sum_{i=1}^{n}R^i_t\right] &=& \rd \left[\left(1-\frac{\theta_k}{n}\right)R^k_t-\theta_kY_t\right]\nonumber\\
		&=& \left[\left(1-\frac{\theta_k}{n}\right)\left(\kappa + \pi^{kT}_t \tilde{\mu}_k   - \frac{1}{2}\pi^{kT}_t\tilde{\sigma}_k\tilde{\sigma}_k^T\pi^{k}_t - \frac{1}{2}\pi^{kT}_t\tilde{\nu}_k\tilde{\nu}_k^T\pi^{k}_t\right)\right.\nonumber\\
		&&\left.-\theta_k\left(\frac{n-1}{n}\kappa + \frac{1}{n}\sum_{i\neq k}\pi^{iT}_t \tilde{\mu}_i   - \frac{1}{2n}\sum_{i\neq k}\pi^{iT}_t\left(\tilde{\sigma}_i\tilde{\sigma}_i^T+\tilde{\nu}_i\tilde{\nu}_i^T\right)\pi^{i}_t\right)\right]\rd t\nonumber\\
		&&+\left[\left(1-\frac{\theta_k}{n}\right)\pi^{kT}_t\tilde{\sigma}_k-\frac{\theta_k}{n}\sum_{i\neq k}\pi^{iT}_t\tilde{\sigma}_i\right]\rd B_t\nonumber\\
		&&+\left(1-\frac{\theta_k}{n}\right)\pi^{kT}_t\tilde{\nu}_k \rd W^k_t - \frac{\theta_k}{n}\sum_{i\neq k}\pi^{iT}_t\tilde{\nu}_i \rd W^i_t.\nonumber
	\end{eqnarray}
	Define two functions:
	\begin{eqnarray}
		H\left(s,\pi\right)\!\!\!&:=&\!\!\!\left(1-\frac{\theta_k}{n}\right)\left(\kappa + \pi^{T} \tilde{\mu}_k   - \frac{1}{2}\pi^{T}\tilde{\sigma}_k\tilde{\sigma}_k^T\pi - \frac{1}{2}\pi^{T}\tilde{\nu}_k\tilde{\nu}_k^T\pi\right)\nonumber\\
		&&-\theta_k\left(\frac{n-1}{n}\kappa + \frac{1}{n}\sum_{i\neq k}\pi^{iT}_s \tilde{\mu}_i   - \frac{1}{2n}\sum_{i\neq k}\pi^{iT}_s\left(\tilde{\sigma}_i\tilde{\sigma}_i^T+\tilde{\nu}_i\tilde{\nu}_i^T\right)\pi^{i}_s\right),\nonumber\\
		G\left(s,\pi\right)\!\!\!&:=&\!\!\!\left[\left(1-\frac{\theta_k}{n}\right)\pi^{T}\tilde{\sigma}_k-\frac{\theta_k}{n}\sum_{i\neq k}\pi^{iT}_s\tilde{\sigma}_i\right]^2\!\!\!\!+\!\left(1-\frac{\theta_k}{n}\right)^2\!\!\pi^{T}\tilde{\nu}_k\tilde{\nu}_k^T\pi +\left(\frac{\theta_k}{n}\right)^2\!\!\sum_{i\neq k}\pi^{iT}_s\tilde{\nu}_i\tilde{\nu}_i^T\pi^{i}_s.\nonumber
	\end{eqnarray}
{ Then, for any fixed initial state $\left(t,r_1,\cdots,r_n\right)\in\left[0,\mathcal{T}\right]\times\mathbb{R}^n$,  we have 
	\begin{eqnarray}
		\mathbb{E}_t\left[\left(1-\frac{\theta_k}{n}\right)R^k_\mathcal{T}-\theta_kY_\mathcal{T}\right] &=& \left(1-\frac{\theta_k}{n}\right)r_k-\theta_k\frac{\sum_{i\neq k}r_i}{n}+\int_{t}^{\mathcal{T}}H\left(s,\pi^k_s\right)\rd s,\nonumber\\
		\mathbb{V}_t\left[\left(1-\frac{\theta_k}{n}\right)R^k_\mathcal{T}-\theta_kY_\mathcal{T}\right] &=&\int_{t}^{\mathcal{T}}G\left(s,\pi^k_s\right)\rd s,\nonumber
	\end{eqnarray}
	and
	\begin{eqnarray}
		J_k\left(\pi^k;\pi^{-k}|\left(t,r_1,\cdots,r_n\right)\right) =  \left(1-\frac{\theta_k}{n}\right)r_k-\theta_k\frac{\sum_{i\neq k}r_i}{n}+\int_{t}^{\mathcal{T}}\left(H\left(s,\pi^k_s\right)-\frac{\gamma_k}{2}G\left(s,\pi^k_s\right)\right)\rd s.\nonumber
	\end{eqnarray} }

For the sake of clarity, the notation $J_k\left(\pi^k;\pi^{-k}|\left(t,r_1,\cdots,r_n\right)\right)$ is simplified to $J_k\left(\pi^k;\pi^{-k}|\left(t,r_k,y\right)\right)$. 

For any admissible strategy $u$ of the $k$-th fund manager and any $h\in\mathbb{R}^{+}$,  $\pi^{k,h,u}$  is as defined in Definition \ref{def2}, and we have:
	\begin{eqnarray}
		J_k\left(\pi^{k,h,u};\pi^{-k}|\left(t,r_k,y\right)\right)\!\!\! &=&\!\!\! J_k\left(\pi^k;\pi^{-k}|\left(t,r_k,y\right)\right) \nonumber\\
		\!\!\!&&\!\!\!+ \int_{t}^{t+h}\!\!\left[\left(H\left(s,u_s\right)-\frac{\gamma_k}{2}G\left(s,u_s\right)\right)\!-\!\left(H\left(s,\pi^k_s\right)-\frac{\gamma_k}{2}G\left(s,\pi^k_s\right)\right)\right]\rd s.\nonumber
	\end{eqnarray}
	If $\left(\pi^{1},\cdots,\pi^{n}\right)$ is a time-consistent equilibrium, form the definition of Definition \ref{def2}, { 
	\begin{eqnarray}
		\pi^k_t=\arg\max_{u}\left\{H\left(t,u\right)-\frac{\gamma_k}{2}G\left(t,u\right)\right\},\quad\forall t \in\left[0,\mathcal{T}\right],\nonumber
	\end{eqnarray} }
	and if $\left(\pi^{1},\cdots,\pi^{n}\right)$ is a constant time-consistent equilibrium, after arranging and omitting the constant term, $\pi^k$ satisfies
	\begin{eqnarray}
		\pi^k=\arg\max_{\pi}&&\left\{\left(1-\frac{\theta_k}{n}\right)\left( \pi^{T} \tilde{\mu}_k   - \frac{1}{2}\pi^{T}\tilde{\sigma}_k\tilde{\sigma}_k^T\pi - \frac{1}{2}\pi^{T}\tilde{\nu}_k\tilde{\nu}_k^T\pi\right)\right.\nonumber\\
		&&\left.-\frac{\gamma_k}{2}\left[\left(1-\frac{\theta_k}{n}\right)\pi^{T}\tilde{\sigma}_k-\frac{\theta_k}{n}\sum_{i\neq k}\pi^{iT}\tilde{\sigma}_i\right]^2-\frac{\gamma_k}{2}\left(1-\frac{\theta_k}{n}\right)^2\pi^{T}\tilde{\nu}_k\tilde{\nu}_k^T\pi\right\},\nonumber
	\end{eqnarray}
	so that
	\begin{eqnarray}
		\tilde{\mu}_k+B\gamma_k\theta_k\tilde{\sigma}_k-\left(1-\frac{\gamma_k\theta_k}{n}+\gamma_k\right)M\pi^k=0,\nonumber
	\end{eqnarray}
	where $B$ and $M$ are given in Eq.~(\ref{B1}) and Eq.~(\ref{M1}).
	Therefore, the maximum point $\pi^{k,*}$ can be organized as
	\begin{eqnarray}
		\pi^{k,*}&=&\frac{1}{1-\frac{\gamma_k\theta_k}{n}+\gamma_k}M^{-1}\left(\tilde{\mu}_k+B\gamma_k\theta_k\tilde{\sigma}_k\right)\nonumber\\
		&=&\frac{1}{1-\frac{\gamma_k\theta_k}{n}+\gamma_k}M^{-1}\left(\tilde{\mu}_k+E\gamma_k\theta_k\tilde{\sigma}_k\right)-\frac{\gamma_k\theta_k}{n\left(1-\frac{\gamma_k\theta_k}{n}+\gamma_k\right)}M^{-1}\tilde{\sigma}_k\tilde{\sigma}_k^T\pi^{k,*},\nonumber
	\end{eqnarray}
	where $E$ has been shown in Eq.~(\ref{E1}).
	Further, let 
	\begin{eqnarray}
		F &:=&\begin{pmatrix}
			\sigma_k^2+\left(1-\frac{\gamma_k\theta_k}{n\left(1+\gamma_k\right)}\right)\nu_k^2 & -\sigma\sigma_k\\-\sigma\sigma_k & \sigma^2
		\end{pmatrix},\nonumber
	\end{eqnarray}
	we have
	\begin{eqnarray}
		\pi^{k,*}&=&\left[\left(1-\frac{\gamma_k\theta_k}{n}+\gamma_k\right)I+\frac{\gamma_k\theta_k}{n}M^{-1}\tilde{\sigma}_k\tilde{\sigma}_k^T\right]^{-1}M^{-1}\left(\tilde{\mu}_k+E\gamma_k\theta_k\tilde{\sigma}_k\right)\nonumber\\
		&=& \left[M\begin{pmatrix}
			1+\gamma_k & \frac{\gamma_k\theta_k	`\sigma_k}{n\sigma}\\0 & 1+\delta_k-\frac{\gamma_k\theta_k}{n}
		\end{pmatrix}\right]^{-1}\left(\tilde{\mu}_k+E\gamma_k\theta_k\tilde{\sigma}_k\right)\nonumber\\
		&=& \frac{1}{\left(1-\frac{\gamma_k\theta_k}{n}+\gamma_k\right)\sigma^2\nu_k^2}F\left(\tilde{\mu}_k+E\gamma_k\theta_k\tilde{\sigma}_k\right).\nonumber
	\end{eqnarray}
Next, consider that each fund manager achieves their optimum, meaning $\left(\pi_1^*,\cdots,\pi_n^*\right)$ forms a constant time-consistent equilibrium, then
	\begin{eqnarray}
		\tilde{\sigma}_k^T\pi^{k,*}=\frac{1}{1+\gamma_k}\frac{\mu}{\sigma}+\frac{\gamma_k\theta_k}{1+\gamma_k}E,\quad k=1,\cdots,n\nonumber
	\end{eqnarray}
	and
	\begin{eqnarray}
		E=\frac{1}{n}\sum_{i=1}^{n}\tilde{\sigma}_i^T\pi^{i,*}=\frac{1}{n}\sum_{i=1}^{n}\frac{1}{1+\gamma_k}\frac{\mu}{\sigma}+\frac{1}{n}\sum_{i=1}^{n}\frac{\gamma_k\theta_k}{1+\gamma_k}E.\nonumber
	\end{eqnarray}
	Noting that $\frac{1}{n}\sum_{i=1}^{n}\frac{\gamma_k\theta_k}{1+\gamma_k}< 1$ , $E$ has the same representation as Eq.~(\ref{K1}), i.e., $E=K$.
	
	At last, the equilibrium is given as Eqs.~(\ref{alp2}) and (\ref{bet2}), and the uniqueness follows from the above proof.
\end{proof}
\section{Proof of Theorem \ref{the4}.}\label{proof4}
\begin{proof}
	{ 
First, for a two-dimensional random variable $\eta$ that is $\mathcal{F}^{MF}_0$-measurable with $\mathbb{E}\|\eta\|^2 < +\infty$, let $\overline{R}_t := \mathbb{E}\left[R^{\eta}_t|\mathcal{F}^{B}_\mathcal{T}\right]$ for $t \in [0, \mathcal{T}]$. Then, the process $\{Z_t\}_{t \in [0, \mathcal{T}]}$ evolves according to Eq.~(\ref{Zt}).

Then,	for any fixed initial state $\left(t,z\right)\in\left[0,\mathcal{T}\right]\times\mathbb{R}$, we have
	\begin{eqnarray}
		\mathbb{E}_t\left[Z_\mathcal{T}\right] &=& z+\int_{t}^{\mathcal{T}}h\left(\pi_s\right)\rd s,\nonumber\\
		\mathbb{V}_t\left[Z_\mathcal{T}\right] &=&\int_{t}^{\mathcal{T}}g\left(\pi_s\right)\rd s,\nonumber\\
	J\left(\pi|\left(t,z\right)\right) &=&  z+\int_{t}^{\mathcal{T}}\left(h\left(\pi_s\right)-\frac{\gamma}{2}g\left(\pi_s\right)\right)\rd s,\nonumber
	\end{eqnarray} }
	where the two functions $h$ and $g$ are defined by
	\begin{eqnarray}
		h\left(\pi\right) &:=& \kappa + \pi^{T} \tilde{\mu}  - \frac{1}{2}\pi^{T}\tilde{\sigma}\tilde{\sigma}^T\pi - \frac{1}{2}\pi^{T}\tilde{\nu}\tilde{\nu}^T\pi-\theta\mathbb{E}\left[\kappa + \eta^{T} \tilde{\mu}   - \frac{1}{2}\eta^{T}\tilde{\sigma}\tilde{\sigma}^T\eta - \frac{1}{2}\eta^{T}\tilde{\nu}\tilde{\nu}^T\eta\right], \nonumber\\
		g\left(\pi\right) &:=& \left(\pi^{T}\tilde{\sigma}-\theta\mathbb{E}\left[\eta^{T}\tilde{\sigma}\right]\right)^2 +  \pi^{T}\tilde{\nu}\tilde{\nu}^T\pi.\nonumber
	\end{eqnarray}
	If $\pi^*$ is a constant time-consistent MFE, form Definition \ref{defination2}, $\eta=\pi^*$ and 
	\begin{eqnarray}
		\pi^*_t\equiv\pi^*=\arg\max_{u}\left\{h\left(u\right)-\frac{\gamma}{2}g\left(u\right)\right\}.\nonumber
	\end{eqnarray}
	Therefore, 
	\begin{eqnarray}
		\pi^*&=&\frac{1}{1+\gamma}\left(\tilde{\sigma}\tilde{\sigma}^T+\tilde{\nu}\tilde{\nu}^T\right)^{-1}\left(\tilde{\mu}+\theta\mathbb{E}\left[\tilde{\sigma}^T\pi^*\right]\tilde{\sigma}\right).\nonumber
	\end{eqnarray}
	Notice that
	\begin{eqnarray}
		\tilde{\sigma}^T\pi^* = \frac{1}{1+\gamma}\tilde{\sigma}^T\left(\tilde{\sigma}\tilde{\sigma}^T+\tilde{\nu}\tilde{\nu}^T\right)^{-1}\left(\tilde{\mu}+\theta\mathbb{E}\left[\tilde{\sigma}^T\pi^*\right]\tilde{\sigma}\right)=\frac{}{1+\gamma}\frac{\mu}{\sigma}+\frac{\theta}{1+\gamma}\mathbb{E}\left[\tilde{\sigma}^T\pi^*\right],\nonumber
	\end{eqnarray}
	then
	{\begin{eqnarray}
		\mathbb{E}\left[\tilde{\sigma}^T\pi^*\right] = \mathbb{E}\left[\frac{1}{1+\gamma}\right]\frac{\mu}{\sigma}+\mathbb{E}\left[\frac{\theta}{1+\delta}\right]\mathbb{E}\left[\tilde{\sigma}^T\pi^*\right].\nonumber
	\end{eqnarray}}
	Thus $\mathbb{E}\left[\tilde{\sigma}^T\pi^*\right] = R\frac{\mu}{\sigma}$,
	where $R$ is shown in Eq.~(\ref{R}). After further analysis, we find that $\alpha^*$ and $\beta^*$ take the forms given in Eqs.~(\ref{alpha2}) and (\ref{beta2}), respectively. The uniqueness then follows from the previous proof.
\end{proof}

\bibliographystyle{apalike}
\bibColoredItems{red}{oecd}
\bibliography{wpref}
\end{document}